\documentclass{article}
\usepackage{url}
\usepackage{natbib}
\usepackage{times}
\usepackage{amsmath,amsfonts,mathrsfs,amsthm}
\usepackage{enumerate}
\usepackage{stackrel}
\newcommand{\N}{\mathbb{N}}
\newcommand{\Q}{\mathbb{Q}}
\newcommand{\R}{\mathbb{R}}

\newcommand{\I}{\mathbb{I}}

\newcommand{\tup}[1]{\langle #1\rangle}
\renewcommand{\cal}[1]{\mathcal{#1}}

\newcommand{\sv}{\varphi}
\renewcommand{\th}{\mathit{th}}
\newcommand{\cosd}{\mathit{cosd}}
\newcommand{\banzhaf}{\beta}
\newcommand{\attvalue}{\chi}
\newcommand{\pattvalue}{\chi^p}
\newcommand{\stateval}{\zeta}

\newcommand{\raw}[1]{\bar{#1}}
\newcommand{\ih}{\stackrel{\mathit{i.h.}}{=}}

\newcommand{\eps}{\varepsilon}

\renewcommand{\vec}[1]{\mathbf{#1}}
\newcommand{\der}[2]{\frac{\partial #1}{\partial #2}}

\newcommand{\vol}{\mathit{Vol}}
\newcommand{\pivotal}{\mathit{Piv}}
\newcommand{\antipivotal}{\mbox{{\it A-Piv}}}
\newcommand{\Piv}{\pivotal}
\newcommand{\APiv}{\antipivotal}

\newtheorem{theorem}{Theorem}[section]
\newtheorem{lemma}[theorem]{Lemma}
\newtheorem{proposition}[theorem]{Proposition}

\newtheorem{corollary}[theorem]{Corollary}
\theoremstyle{definition}

\newtheorem{example}[theorem]{Example}
\newtheorem{definition}[theorem]{Definition}

\usepackage{xcolor}
\newcount\Comments  
\newcommand{\kibitz}[2]{\ifnum\Comments=1{\color{#1}{\sf #2}}\fi}			
\newcommand{\ariel}[1]{\kibitz{red}{[Ariel: #1]}}
\newcommand{\yair}[1]{\kibitz{green!60!black}{[Yair: #1]}}							

\newcount\Appendix
\title{Influence in Classification via Cooperative Game Theory}
\author{Amit Datta, Anupam Datta, Ariel D. Procaccia and Yair Zick\\Carnegie Mellon University\\
\texttt{amitdatta,danupam@cmu.edu}\\\texttt{arielpro,yairzick@cs.cmu.edu}}
\date{}
\begin{document}
\maketitle
\begin{abstract}
A dataset has been classified by some unknown classifier into two types of points. What were the most important factors in determining the classification outcome? In this work, we employ an axiomatic approach in order to uniquely characterize an influence measure: a function that, given a set of classified points, outputs a value for each feature  corresponding to its influence in determining the classification outcome. We show that our influence measure takes on an intuitive form when the unknown classifier is linear. 
Finally, we employ our influence measure in order to analyze the effects of user profiling on Google's online display advertising. 
\end{abstract}
\section{Introduction}
A recent white house report~\citep{whitehouse2014bigdata} highlights some of the major risks in the ubiquitous use of big data technologies. According to the report, one of the major issues with large scale data collection and analysis is a glaring lack of transparency. For example, a credit reporting company collects consumer data from third parties, and uses machine learning analysis to estimate individuals' credit score. On the one hand, this method is ``impartial'': an emotionless algorithm cannot be accused of being malicious (discriminatory behavior is not hard-coded). However, it is hardly transparent; indeed, it is difficult to tease out the determinants of one's credit score: it depends on the user's financial activities, age, address, the behavior of similar users and many other factors. This is a major issue: big-data analysis does not intend to discriminate, but inadvertent discrimination does occur: treating users differently based on unfair criteria (e.g. online retailers offering different discounts or goods based on place of residence or past purchases). 

In summary, big data analysis leaves users vulnerable. They may be discriminated against, and no one (including the algorithm's developers!) may even know why; what's worse, traditional methods for preserving user anonymity (e.g. by ``opting out'' of data collection) offer little protection; big data techniques allow companies to infer individuals' data based on similar users~\citep{barocas2014bigdata}. Since it is often difficult to ``pop the hood'' and understand the inner workings of classification algorithms, maintaining transparency in classification is a major challenge. In more concrete terms, transparency can be interpreted as understanding what influences the decisions of a black-box classifier. This is where our work comes in. 

Suppose that we are given a dataset $B$ of users; here, every user $\vec a \in B$ can be thought of as a vector of features (e.g. $\vec a = (\mbox{age, gender, IP address}\dots)$), where the $i$-th coordinate of $\vec a$ corresponds to the state of the $i$-th feature. Each $\vec a$ has a value $v(\vec a)$ (say, the credit score of $\vec a$). We are interested in the following question: 
{\em given a dataset $B$ of various feature vectors and their values, how influential was each feature in determining these values?} 

In more detail, given a set $N = \{1,\dots,n\}$ of features, a dataset $B$ of feature profiles, where every profile $\vec a$ has a value $v(\vec a)$, we would like to compute a measure $\phi_i(N,B,v)$ that corresponds to feature $i$'s importance in determining the labels of the points in $B$. 
We see this work as an important first step towards a concrete methodology for transparency analysis of big-data algorithms.

\paragraph{Our Contribution:} We take an axiomatic approach --- which draws heavily on cooperative game theory --- to define an influence measure. The merit of our approach lies in its independence of the underlying structure of the classification function; all we need is to collect data on its behavior.  

We show that our influence measure is the unique measure satisfying some natural properties (Section~\ref{sec:featureinfluence}). As a case study, we show that when the input values are given by a linear classifier, our influence measure has an intuitive geometric interpretation (Section~\ref{sec:linear}). Finally, we show that our axioms can be extended in order to obtain other influence measures (Section~\ref{sec:extensions}). For example, our axioms can be used to obtain a measure of {\em state influence}, as well as influence measures where a prior distribution on the data is assumed, or a measure that uses pseudo-distance between user profiles to measure influence. 

We complement our theoretical results with an implementation of our approach, which serves as a proof of concept (Section~\ref{sec:ads}). Using our framework, we identify ads where certain user features have a significant influence on whether the ad is shown to users. Our experiments show that our influence measures behave in a desirable manner. 
In particular, a Spanish language ad --- clearly biased towards Spanish speakers --- demonstrated the highest influence of any feature among all ads.

\subsection{Related Work}
 Axiomatic characterizations have played an important role in the design of provably fair revenue divisions~\citep{shapleyvalue,shapleyyoung,banzhaf,lehrer1988axiomatization}. Indeed, one can think of the setting we describe as a generalization of cooperative games, where agents can have more than one state --- in cooperative games, agents are either present or absent from a coalition. 
Some papers extend cooperative games to settings where agents have more than one state, and define influence measures for such settings~\citep{ocfgeb,zick2014ocfarb}; however, our setting is far more general. 

Our definition of influence measures the ability of a feature to affect the classification outcome if changed (e.g. how often does a change in gender cause a change in the display frequency of an ad); this idea is used in the analysis of cause~\citep{halpern2005causes,tian2000probabilities}, and responsibility~\citep{chockler2004responsibility}; our influence measure can be seen as an application of these ideas to a classification setting. 

Influence measures are somewhat related to {\em feature selection}~\citep{blum1997selection}. 
Feature selection is the problem of finding the set of features that are most relevant to the classification task, 
in order to improve the performance of a classifier on the data; 
that is, it is the problem of finding a subset of features, 
such that if we train a classifier using just those features, 
the error rate is minimized. Some of the work on feature selection employs feature ranking methods; some even use the Shapley value as a method for selecting the most important features~\citep{cohen2005feature}. 
Our work differs from feature selection both in its objectives and its methodology. Our measures can be used in order to rank features, but we are not interested in training classifiers; rather, we wish to decide which features influence the decision of an unknown classifier. That said, one can certainly employ our methodology in order to rank features in feature selection tasks. 

When the classifier is linear, our influence measures take on a particularly intuitive interpretation as the aggregate volume between two hyperplanes~\citep{marichal2006slices}.

Recent years have seen tremendous progress on methods to enhance fairness in classification~\citep{DworkHPRZ12,kamishima2011fairness}, user privacy~\citep{Balebako'12,DBLP:conf/kdd/PedreschiRT08,Wills'12} and the prevention of  discrimination~\citep{4909197,Calders,luong2011k}. Our work can potentially inform all of these research thrusts: a classifier can be deemed fair if the influence of certain features is low; for example, high gender influence may indicate discrimination against a certain gender. In terms of privacy, if a hidden feature (i.e. one that is not part of the input to the classifier) has high influence, this indicates a possible breach of user privacy.

\section{Axiomatic Characterization}
\label{sec:featureinfluence}
We begin by briefly presenting our model. Given a set of {\em features} $N = \{1,\dots,n\}$, let $A_i$ be the set of possible {\em values}, or {\em states} that feature $i$ can take; for example, the $i$-th feature could be gender, in which case $A_i = \{\mbox{male},\mbox{female},\mbox{other}\}$. 
We are given {\em partial} outputs of a function over a dataset containing feature profiles. That is, we are given a subset $B$ of $A = \prod_{i \in N}A_i$, and a valuation $v(\vec a)$ for every $\vec a \in B$. By given, we mean that we do not know the actual structure of $v$, but we know what values it takes over the dataset $B$. 
Formally, our input is a tuple $\cal G = \tup{N,B,v}$, where $v:A\to \Q$ is a function assigning a value of $v(\vec a)$ to each data point $\vec a \in B$. We refer to $\cal G$ as the {\em dataset}. When $v(\vec a) \in \{0,1\}$ for all $\vec a \in B$, $v$ is a {\em binary classifier}. When $B = A$ and $|A_i| = 2$ for all $i \in N$, the dataset corresponds to a standard TU cooperative game~\citep{compcoopbook} (and is a simple game if $v(\vec a) \in \{0,1\}$). 

We are interested in answering the following question: {\em how influential is feature $i$?}
Our desired output is a measure $\phi_i(\cal G)$ that will be associated with each feature $i$. The measure $\phi_i(\cal G)$ should be a good metric of the importance of $i$ in determining the values of $v$ over $B$. 

Our goal in this section is to show that there exists a unique influence measure that satisfies certain natural axioms. We begin by describing the axioms, starting with symmetry.

 Given a dataset $\cal G  = \tup{N,B,v}$ and a bijective mapping $\sigma$ from $N$ to itself, we define $\sigma \cal G = \tup{\sigma N,\sigma B,\sigma v}$ in the natural way: $\sigma N$ has all of the features relabeled according to $\sigma$ (i.e. the index of $i$ is now $\sigma(i)$); $\sigma B$ is $\{\sigma \vec a \mid \vec a \in B\}$, and $\sigma v(\sigma \vec a) = v(\vec a)$ for all $\sigma \vec a \in \sigma B$. Given a bijective mapping $\tau:A_i \to A_i$ over the states of some feature $i \in N$, we define $\tau \cal G = \tup{N,\tau B,\tau v}$ in a similar manner. 

\begin{definition}\label{def:symmetry}

An influence measure $\phi$ satisfies the {\em feature symmetry} property if it is invariant under relabelings of features: given a dataset $\cal G = \tup{N,B,v}$ and some bijection $\sigma:N \to N$, $\phi_i(\cal G) = \phi_{\sigma(i)}(\sigma \cal G)$ for all $i \in N$. 
A influence measure $\phi$ satisfies the {\em state symmetry} property if it is invariant under relabelings of states: given a dataset $\cal G = \tup{N,B,v}$, some $i \in N$, and some bijection $\tau:A_i \to A_i$, $\phi_j(\cal G) = \phi_j(\tau\cal G)$ for all $j \in N$. Note that it is possible that $i \ne j$. A measure satisfying both state and feature symmetry is said to satisfy the {\em symmetry} axiom (Sym).
\end{definition}
Feature symmetry is a natural extension of the symmetry axiom defined for cooperative games (see e.g.~\citep{banzhaf,lehrer1988axiomatization,shapleyvalue}). However, state symmetry does not make much sense in classic cooperative games; it would translate to saying that for any set of players $S \subseteq N$ and any $j \in N$, the value of $i$ is the same if we treat $S$ as $S\setminus \{j\}$, and $S\setminus \{j\}$ as $S$. While in the context of cooperative games this is rather uninformative, we make non-trivial use of it in what follows. 

We next describe a sufficient condition for a feature to have no influence: a feature should not have any influence if it does not affect the outcome in any way. Formally, a feature $i \in N$ is a {\em dummy} if $v(\vec a) = v(\vec a_{-i},b)$ for all $\vec a \in B$, and all $b \in A_i$ such that $(\vec a_{-i},b) \in B$. 
\begin{definition}\label{def:dummy}
An influence measure $\phi$ satisfies the {\em dummy} property if $\phi_i(\cal G) = 0$ whenever $i$ is a dummy in the dataset $\cal G$.
\end{definition}
The dummy property is a standard extension of the dummy property used in value characterizations in cooperative games. However, when dealing with real datasets, it may very well be that there is no vector $\vec a \in B$ such that $(\vec a_{-i},b) \in B$; this issue is discussed further in Section~\ref{sec:conclusions}.

Cooperative game theory employs a notion of value additivity in the characterization of both the Shapley and Banzhaf values. Given two datasets $\cal G_1 = \tup{N,B,v_1},\cal G_2= \tup{N,B,v_2}$, we define $\cal G = \tup{N,A,v} = \cal G_1+ \cal G_2$ with $v(\vec a) = v_1(\vec a) + v_2(\vec a)$ for all $\vec a \in B$. 
\ariel{Didn't we restrict $v$ to be a Boolean-valued function?}
\yair{now we don't}
\begin{definition}\label{def:additivity}
An influence measure $\phi$ satisfies additivity (AD) if $\phi_i(\cal G_1+ \cal G_2) = \phi_i(\cal G_1) + \phi_i(\cal G_2)$ for any two datasets $\cal G_1 = \tup{N,B,v_1},\cal G_2 = \tup{N,B,v_2}$. 
\end{definition}
The additivity axiom is commonly used in the axiomatic analysis of revenue division in cooperative games (see~\citep{lehrer1988axiomatization,shapleyvalue}); however, it fails to capture a satisfactory notion of influence in our more general setting.
We now show that any measure that satisfies additivity, in addition to the symmetry and dummy properties, must evaluate to 
zero for all features. To show this, we first define the following simple class of datasets. 
\begin{definition}
Let $\cal U_{\vec a} = \tup{N,A,u_{\vec a}}$ be the dataset defined by the classifier $u_{\vec a}$, where $u_{\vec a}(\vec a') = 1$ if $\vec a' = \vec a$, and is 0 otherwise.
The dataset $\cal U_{\vec a}$ is referred to as the {\em singleton dataset} over $\vec a$.
\end{definition}
It is an easy exercise to show that additivity implies that for any scalar $\alpha \in \Q$, $\phi_i(\alpha\cal G) = \alpha\phi_i(\cal G)$, where the dataset $\alpha\cal G$ has the value of every point scaled by a factor of $\alpha$.
\begin{proposition}\label{prop:impossibility1}
Any influence measure that satisfies the (Sym), (D) and (AD) axioms evaluates to 
zero for all features.
\end{proposition}
\begin{proof}
First, we show that for any $\vec a,\vec a' \in A$ and any $b \in A_i$, it must be the case that $\phi_i(\cal U_{(\vec a_{-i},b)}) = \phi_i(\cal U_{(\vec a_{-i}',b)})$. This is true because we can define a bijective mapping from $\cal U_{(\vec a_{-i},b)}$ to $\cal U_{(\vec a_{-i}',b)}$: for every $j \in N\setminus\{i\}$, we swap $a_j$ and $a_j'$. By state symmetry, $\phi_i(\cal U_{(\vec a_{-i},b)}) =\phi_i(\cal U_{(\vec a_{-i}',b)})$. 

Next, if $\phi$ is additive, then for any dataset $\cal G = \tup{N,B,v}$, $\phi_i(\cal G) = \sum_{\vec a \in B}v(\vec a)\phi_i(\cal U_{\vec a})$.
That is, the influence of a feature must be the sum of its influence
over singleton datasets, scaled by $v(\vec a)$. 

Now, suppose for contradiction that there exists some singleton dataset $\cal U_{\bar{\vec a}}$ ($\bar{\vec a} \in B$) for which some feature $i \in N$ does not have an influence of zero. 
That is, we assume that $\phi_i(\cal U_{\bar{\vec a}}) \ne 0$. 
We define a dataset $\cal G = \tup{N,A,v}$ in the following manner: 
for all $\vec a \in A$ such that $\vec a_{-i} = \bar{\vec a}_{-i}$, we set $v(\vec a) = 1$, and $v(\vec a) = 0$ if $\vec a_{-i} \ne \bar{\vec a}_{-i}$. 
In the resulting dataset, $v(\vec a)$ is solely determined by the values of features in $N\setminus\{i\}$; 
in other words $v(\vec a) = v(\vec a_{-i},b)$ for all $b \in A_i$, hence feature $i$ is a dummy. According to the dummy axiom, we must have that $\phi_i(\cal G) = 0$; however, 
\begin{align*}
0&=\phi_i(\cal G) 	=  \sum_{\vec a:v(\vec a) = 1}\phi_i(\cal U_{\vec a})= \sum_{b \in A_i}\phi_i(\cal U_{(\bar{\vec a}_{-i},b)})\\
								&=  \sum_{b \in A_i}\phi_i(\cal U_{\bar{\vec a}}) = |A_i|\phi_i(\cal U_{\bar{\vec a}})> 0,
\end{align*}
where the first equality follows from the decomposition of $\cal G$ into singleton datasets,
and the third equality holds by Symmetry. \ariel{Added this explanation, please check.}\yair{makes sense} This is a contradiction.
\end{proof}

As Proposition~\ref{prop:impossibility1} shows, the additivity, symmetry and dummy properties do not lead to a meaningful description of influence. A reader familiar with the axiomatic characterization of the Shapley value~\citep{shapleyvalue} will find this result rather disappointing: the classic characterizations of the Shapley and Banzhaf values assume additivity (that said, The axiomatization by \cite{shapleyyoung} does not assume additivity). 

\ariel{In the following we definitely need $v$ to be Boolean-valued. Also, the paragraph is confusing because every vector is either ``winning'' or ``losing'', but that's not the case in the def of disjoint union.} 
We now show that there is an influence measure uniquely defined by an alternative axiom, which echoes the union intersection property described by \cite{lehrer1988axiomatization}. 
In what follows, we assume that all datasets are classified by a binary classifier.
We write $W(B)$ to be the set of all profiles in $B$ such that $v(\vec a) = 1$, and $L(B)$ to be the set of all profiles in $B$ that have a value of 0. We refer to $W(B)$ as the {\em winning profiles} in $B$, and to $L(B)$ as the {\em losing profiles} in $B$. We can thus write $\phi_i(W(B),L(B))$, rather than $\phi_i(\cal G)$. Given two disjoint sets $W,L \subseteq A$, we can define the dataset as $\cal G = \tup{W,L}$, and the influence of $i$ as $\phi_i(W,L)$, without explicitly writing $N,B$ and $v$.
As we have seen, no measure can satisfy the additivity axiom (as well as symmetry and dummy axioms) without being trivial. We now propose an alternative influence measure, captured by the following axiom: 
\begin{definition}\label{def:disjoint-union}
An influence measure $\phi$ satisfies the {\em disjoint union (DU)} property if for any $Q \subseteq A$, and any disjoint $R,R' \subseteq A\setminus Q$, $\phi_i(Q,R) + \phi_i(Q,R') = \phi_i(Q,R\cup R')$, and $\phi_i(R,Q) + \phi_i(R',Q) = \phi_i(R\cup R',Q)$.
\end{definition}
An influence measure $\phi$ satisfying the (DU) axiom is additive with respect to independent observations {\em of the same type}.
Suppose that we are given the outputs of a binary classifier on two datasets: $\cal G_1 = \tup{W,L_1}$ and $\cal G_2 = \tup{W,L_2}$. The (DU) axiom states that the ability of a feature to affect the outcome on $\cal G_1$ is independent of its ability to affect the outcome in $\cal G_2$, if the winning states are the same in both datasets. 


Replacing additivity with the disjoint union property yields a unique influence measure, with a rather simple form.
\begin{align}
\attvalue_i(\cal G) = \sum_{\vec a \in B}\sum_{b \in A_i:(\vec a_{-i},b) \in B}|v(\vec a_{-i},b) - v(\vec a)|\label{eq:cause}
\end{align}
$\attvalue$ measures the number of times that a change in the state of $i$ causes a change in the classification outcome. If we normalize $\attvalue$ and divide by $|B|$, the resulting measure has the following intuitive interpretation: pick a vector $\vec a \in B$ uniformly at random, and count the number of points in $A_i$ for which $(\vec a_{-i},b)\in B$ and $i$ changes the value of $\vec a$. We note that when all features have two states and $B = A$, $\attvalue$ coincides with the (raw) Banzhaf power index~\citep{banzhaf}.

We now show that $\attvalue$ is a unique measure satisfying (D), (Sym) and (DU). We begin by presenting the following lemma, which characterizes influence measures satisfying (D), (Sym) and (DU) when dataset contains only a single feature. 

\begin{lemma}\label{lem:symmetry-dependence-on-winning-states}
Let $\phi$ be an influence measure that satisfies state symmetry, and let $\cal G_1 = \tup{\{i\},A_i,v_1}$ and $\cal G_2 = \tup{\{i\},A_i,v_2}$ be two datasets with a single feature $i$; if the number of winning states under $\cal G_1$ and $\cal G_2$ is identical, then $\phi_i(\cal G_1) = \phi_i(\cal G_2)$.
\end{lemma}
\yair{replaced proof with sketch}
\begin{proof}[Proof Sketch]
We simply construct a bijective mapping from the winning states of $i$ under $\cal G_1$ and its winning states in $\cal G_2$. By state symmetry, $\phi_i(\cal G_1) = \phi_i(\cal G_2)$.
\end{proof}
Lemma~\ref{lem:symmetry-dependence-on-winning-states} implies that for single feature games, the value of a feature only depends on the number of winning states, rather than their identity.

We are now ready to show the main theorem for this section: $\attvalue$ is the unique influence measure satisfying the three axioms above, up to a constant factor. 
\begin{theorem}\label{thm:attvalue}
An influence measure $\phi$ satisfies (D), (Sym) and (DU) if and only if there exists a constant $C$ such that for every dataset $\cal G = \tup{N,B,v}$
$$\phi_i(\cal G) = C\cdot \attvalue_i(\cal G).$$
\end{theorem}
\begin{proof}
It is an easy exercise to verify that $\attvalue$ satisfies the three axioms, so we focus on the ``only if'' direction.

We present our proof assuming that we are given the set $A$ as data; the proof goes through even if we assume that we are presented with some arbitrary $B \subseteq A$. Let us write $W = W(A)$ and $L = L(A)$. Given some $\vec a_{-i} \in A_{-i}$, we write $L_{\vec a_{-i}} = \{\bar{\vec a} \in L\mid \vec a_{-i} = \bar{\vec a}_{-i}\}$, and $W_{\vec a_{-i}} = \{\bar{\vec a} \in W\mid \vec a_{-i} = \bar{\vec a}_{-i}\}$.

Using the disjoint union property, we can decompose $\phi_i(W,L)$ as follows: 
\begin{align}\label{eq:state-additivity}
\phi_i(W,L) = \sum_{\vec a_{-i} \in A_{-i}}\sum_{\bar{\vec a}_{-i}\in A_{-i}}\phi_i(W_{\vec a_{-i}},L_{\bar{\vec a}_{-i}}).
\end{align}
Now, if $\bar{\vec a}_{-i} \ne \vec a_{-i}$, then feature $i$ is a dummy given the dataset provided. Indeed, state profiles are either in $W_{\vec a_{-i}}$ or in $L_{\bar{\vec a}_{-i}}$; that is, if $v(\vec a_{-i},b) = 0$, then $(\vec a_{-i},b)$ is unobserved, and if $v(\bar{\vec a}_{-i},b) = 1$, then $(\bar{\vec a}_{-i},b)$ is unobserved. We conclude that 
\begin{align}\label{eq:state-additivity2}
\phi_i(W,L) = \sum_{\vec a_{-i} \in A_{-i}}\phi_i(W_{\vec a_{-i}},L_{\vec a_{-i}}).
\end{align}
Let us now consider $\phi_i(W_{\vec a_{-i}},L_{\vec a_{-i}})$. 
Since $\phi$ satisfies state symmetry, Lemma~\ref{lem:symmetry-dependence-on-winning-states} implies that $\phi_i$ can only possibly depend on $\vec a_{-i}$, $|W_{\vec a_{-i}}|$ and $|L_{\vec a_{-i}}|$. Next, for any $\vec a_{-i}$ and $\vec a_{-i}'$ such that $|L_{\vec a_{-i}}| = |L_{\vec a_{-i}'}|$ and $|W_{\vec a_{-i}}| = |W_{\vec a_{-i}'}|$, so by Lemma~\ref{lem:symmetry-dependence-on-winning-states} $\phi_i(W_{\vec a_{-i}},L_{\vec a_{-i}}) = \phi_i(W_{\vec a_{-i}'},L_{\vec a_{-i}'})$. In other words $\phi_i$ only depends on $|W_{\vec a_{-i}}|,|L_{\vec a_{-i}}|$, and not on the identity of $\vec a_{-i}$.

Thus, one can see $\phi_i$ for a single feature as a function of two parameters, $w$ and $l$ in $\N$, where $w$ is the number of winning states and $l$ is the number of losing states. According to the dummy property, we know that $\phi_i(w,0) = \phi_i(0,l) = 0$; moreover, the disjoint union property tells us that $\phi_i(x,l) + \phi_i(y,l) = \phi_i(x+y,l)$, and that $\phi_i(w,x) + \phi_i(w,y) = \phi_i(w,x+y)$. We now show that $\phi_i(w,l) = \phi_i(1,1)wl$. 

Our proof is by induction on $w+l$. For $w+l = 2$ the claim is clear. Now, assume without loss of generality that $w > 1$ and $l \ge 1$; then we can write $w = x+y$ for $x,y\in \N$ such that $1\le x,y <w$. By our previous observation, 
\begin{align*}
\phi_i(w,l) &=		 \phi_i(x,l) + \phi_i(y,l) \\ 
						&\ih  \phi_i(1,1)xl + \phi_i(1,1)yl = \phi_i(1,1)wl\nonumber .
\end{align*}
Now, $\phi_i(1,1)$ is the influence of feature $i$ when there is exactly one losing state profile, and one winning state profile. 
We write $\phi_i(1,1)= c_i$. 

Let us write $W_i(\vec a_{-i})= \{b \in A_i\mid v(\vec a_{-i},b) = 1\}$ and $L_i(\vec a_{-i}) = A_i\setminus W_i(\vec a_{-i})$. Thus, $|W_{\vec a_{-i}}| = |W_i(\vec a_{-i})|$, and $|L_{\vec a_{-i}}| = |L_i(\vec a_{-i})|$. 
Putting it all together, we get that 
\begin{align}\label{eq:winloss}
\phi_i(\cal G) = c_i\sum_{\vec a_{-i} \in A_{-i}} |W_i(\vec a_{-i})|\cdot |L_i(\vec a_{-i})|
\end{align}
We just need to show that the measure given in~\eqref{eq:winloss} equals $\attvalue_i$ (modulo $c_i$). Indeed, \eqref{eq:winloss} equals $\sum_{\vec a \in A:\  v(\vec a) = 0} |W_i(\vec a_{-i})|$, which in turn equals $\sum_{\substack{\vec a \in A:\ v(\vec a) = 0}}\ \sum_{b \in A_i} |v(\vec a_{-i},b) - v(\vec a)|$. Similarly, \eqref{eq:winloss} equals $$\sum_{\substack{\vec a \in A: v(\vec a) = 1}}\ \sum_{b \in A_i} |v(\vec a_{-i},b) - v(\vec a)|.$$ Thus,
$$\sum_{\vec a_{-i} \in A_{-i}} |W_i(\vec a_{-i})|\cdot |L_i(\vec a_{-i})| = \frac12\sum_{\vec a \in A}\sum_{b \in A_i}|v(\vec a_{-i},b) - v(\vec a)|;$$
in particular, for every dataset $\cal G=\tup{N,A,v}$ and every $i \in N$, there is some constant $C_i$ such that $\phi_i(\cal G) = C_i\attvalue_i(\cal G)$.
To conclude the proof, we must show that $C_i= C_j$ for all $i,j \in N$. Let $\sigma:N\to N$ be the bijection that swaps $i$ and $j$; then $\phi_i(\cal G) = \phi_{\sigma(i)}(\sigma\cal G)$. By feature symmetry, $C_i\attvalue_i(\cal G) = \phi_i(\cal G) = \phi_{\sigma(i)}(\sigma\cal G) = \phi_j(\sigma\cal G) = C_j\attvalue_j(\sigma\cal G) = C_j\attvalue_i(\cal G)$, thus $C_i = C_j$. 
\end{proof}

\section{Case Study: Influence for Linear Classifiers}\label{sec:linear}
To further ground our results, we now present their application to the class of linear classifiers. For this class of functions, our influence measure takes on an intuitive interpretation. 

A {\em linear classifier} is defined by a hyperplane in $\R^n$; all points that are on one side of the hyperplane are colored blue (in our setting, have value 1), and all points on the other side are colored red (have a value of 0). Formally, we associate a weight $w_i \in \R$ with every one of the features in $N$ (we assume that $w_i \ne 0$ for all $i \in N$); a point $\vec x \in \R^n$ is blue if $\vec x \cdot \vec w \ge q$, where $q \in \R$ is a given parameter. 
The classification function $v:\R^n \to \{0,1\}$ is given by 
\begin{align}\label{eq:linear-classifier}
v(\vec x) = \begin{cases}1 &\mbox{if } \vec x\cdot \vec w \ge q\\0 &\mbox{otherwise.} \end{cases}
\end{align}

Fixing the value of $x_i$ to some $b \in \R$, let us consider the set 
$W_i(b) = \{\vec x_{-i} \in \R^{n-1}\mid v(\vec x_{-i},b) =1\}$; we observe that if $b < b'$ and $w_i > 0$, then $W_i(b) \subset W_i(b')$ (if $w_i < 0$ then $W_i(b')\subset W_i(b)$). 
Given two values $b,b' \in \R$, we denote by $$D_i(b,b') = \{\vec x_{-i} \in \R^{n-1}\mid v(\vec x_{-i},b) \ne v(\vec x_{-i},b')\}.$$ By our previous observation, if $b < b'$ then $D_i(b,b') = W_i(b') \setminus W_i(b)$, and if $b > b'$ then $D_i(b,b') = W_i(b)\setminus W_i(b')$. 

Suppose that rather than taking points in $\R^n$, we only take points in $[0,1]^n$; then we can define $|D_i(b,b')| = \vol(D_i(b,b'))$, where $$\vol(D_i(b,b')) = \int_{\vec x_{-i} \in [0,1]^{n-1}} |v(\vec x_{-i},b') - v(\vec x_{-i},b)|\partial \vec x_{-i}.$$
In other words, in order to measure the total influence of setting the state of feature $i$ to $b$, we must take the total volume of $D_i(b,b')$ for all $b' \in [0,1]$, which equals $\int_{b' = 0}^1\vol(D_i(b,b'))\partial b$. Thus, the total influence of setting the state of $i$ to $b$ is $\int_{\vec x \in [0,1]^{n}}|v(\vec x_{-i},b) - v(\vec x)|\partial \vec x$.
The total influence of $i$ would then be naturally the total influence of its states, i.e.
\begin{align}\label{eq:volume-feature-influence}
\int_{b = 0}^1\int_{\vec x \in [0,1]^{n}}|v(\vec x_{-i},b) - v(\vec x)|\partial \vec x \partial b.
\end{align}
The formula in Equation~\eqref{eq:volume-feature-influence} is denoted by $\attvalue_i(\vec w;q)$. Equation~\eqref{eq:cause} is a discretized version of Equation~\eqref{eq:volume-feature-influence}; the results of Section~\ref{sec:featureinfluence} can be extended to the continuous setting, with only minimal changes to the proofs. 

We now show that the measure given in~\eqref{eq:volume-feature-influence} agrees with the weights in some natural manner. 
This intuition is captured in Theorem~\ref{thm:monotone-weight} (proof omitted).

\begin{theorem}\label{thm:monotone-weight}
Let $v$ be a linear classifier defined by $\vec w$ and $q$; then $\attvalue_i(\cal G) \ge \attvalue_j(\cal G)$ if and only if $|w_i| \ge |w_j|$.
\end{theorem}

Given Theorem~\ref{thm:monotone-weight}, one would expect the following to hold: suppose that we are given two weight vectors, $\vec w, \vec w' \in \R^n$ such that $w_j= w_j'$ for all $j \ne i$, but $w_i < w_i'$. Let $v$ be the linear classifier defined by $\vec w$ and $q$ and $v'$ be the linear classifier defined by $\vec w'$ and $q$. Is it the case that feature $i$ is more influential under $v'$ than under $v$? In other words, does influence monotonicity hold when we increase the weight of an individual feature? The answer to this is negative. 
\begin{example}\label{ex:counterex-monotone}
Let us consider a single feature game where $N = \{1\}$, $A_1 = [0,1]$, and $v(x) = 1$ if $wx \ge q$, and $v(x) = 0$ if $wx < q$ for a given $w > q$. The fraction of times that $1$ is pivotal is $$|\Piv_1| = \int_{b = 0}^1\int_{x = 0}^1 \I(v(b) {=} 1 \land v(x) {=} 0)\partial x \partial b;$$ simplifying, this expression is equal to $\left(1 - \frac qw\right)\frac qw$.
We can show that $\attvalue_1 = 2|\Piv_1|$ \ariel{Not sure the factor of 2 is needed}\yair{we are not counting all of the times that $v(b) = 0\land v(x) = 1$, i.e. when $b$ is ``anti-pivotal''; but, the two sets are isomorphic.}, we have that $\attvalue_1$ is maximized when $q = 2w$; in particular, $\attvalue_1$ is monotone increasing when $q < w \le 2q$, and it is monotone decreasing when $w \ge 2q$.
\end{example}
\yair{Do you guys think that this discussion is good? I think it is good to highlight different properties of our measure, the question is whether we are being too verbose}
Example~\ref{ex:counterex-monotone} highlights the following phenomenon: fixing the 
other features to be $\vec a_{-i}$, the influence of $i$ is maximized when $|L_{\vec a_{-i}}| = |W_{\vec a_{-i}}|$. This can be interpreted probabilistically: we sample a random feature from $B$, and assume that for any fixed $\vec a_{-i} \in A_{-i}$, $\Pr[v(\vec a_{-i},b) = 1] = \frac12$. The better a feature $i$ agrees with our assumption, the more $i$ is rewarded. More generally, an influence measure satisfies the {\em agreement with prior assumption} (APA) axiom if for any vector $(p_1,\dots,p_n) \in [0,1]^n$, and any fixed $\vec a_{-i} \in A_{-i}$, $i$'s influence increases as $|\Pr[v(\vec a_{-i},b) = 1] - p_i|$ decreases. 
A variant of the symmetry axiom (that reflects changes in probabilities when labels change), along with the dummy and disjoint union axioms can give us a weighted influence measure as described in Section~\ref{sec:weighted-influence}, that also satisfies the (APA) axiom.

\section{Extensions of the Feature Influence Measure}\label{sec:extensions}
Section~\ref{sec:featureinfluence} presents an axiomatic characterization of feature influence, where the value of each feature vector is either zero or 1. We now present a few possible extensions of the measure, and the variations on the axioms that they require.
\subsection{State Influence}
Section~\ref{sec:featureinfluence} provided an answer to questions of the following form: what is the impact of gender on classification outcomes? 
The answer provided in previous sections was that influence was a function of the feature's ability to change outcomes by changing its state. 

It is also useful to ask a related question: what is the impact of the gender feature being set to ``female'' on classification outcomes? 
In other words, rather than measuring feature influence, we are measuring the influence of feature $i$ being in a certain {\em state}. The results described in Section~\ref{sec:featureinfluence} can be easily extended to this setting. 
Moreover, the impossibility result described in Proposition~\ref{prop:impossibility1} no longer holds when we measure state --- rather than feature --- influence: we can replace the disjoint union property with additivity to obtain an alternative classification of state influence. 
\subsection{Weighted Influence}\label{sec:weighted-influence}
Suppose that in addition to the dataset $B$, we are given a weight function $w:B \to \R$. $w(\vec a)$ can be thought of as the number of occurrences of the vector $\vec a$ in the dataset, the probability that $\vec a$ appears, or some intrinsic importance measure of $\vec a$. Note that in Section~\ref{sec:featureinfluence} we implicitly assume that all points occur at the same frequency (are equally likely) and are equally important. 
A simple extension of the disjoint union and symmetry axioms to a weighted variant shows that the only weighted influence measure that satisfies these axioms is 
$$\attvalue_i^w(B) =\sum_{\vec a \in B}\sum_{b \in A_i:(\vec a_{-i},b) \in B}w(\vec a)|v(\vec a_{-i},b) - v(\vec a)|.$$
\subsection{General Distance Measures}\label{sec:distancemeasures3}
Suppose that instead of a classifier $v:A \to \{0,1\}$ we are given a pseudo-distance measure: that is, a function $d:A\times A \to \R$ that satisfies $d(\vec a,\vec a') = d(\vec a',\vec a)$, $d(\vec a,\vec a) = 0$ and the triangle inequality. Note that it is possible that $d(\vec a,\vec a') = 0$ but $\vec a \ne \vec a'$. An axiomatic analysis in such general settings is possible, but requires more assumptions on the behavior of the influence measure. Such an axiomatic approach leads us to show that the influence measure  
$$\attvalue_i^d(B) = \sum_{\vec a \in B}\sum_{b \in A_i:(\vec a_{-i},b) \in B}d((\vec a_{-i},b),\vec a)$$
is uniquely defined via some natural axioms. The additional axioms are a simple extension of the disjoint union property, and a minimal requirement stating that when $B = \{\vec a,(\vec a_{-i},b)\}$, then the influence of a feature is $\alpha d((\vec a_{-i},b),\vec a)$ for some constant $\alpha$ independent of $i$. The extension to pseudo-distances proves to be particularly useful when we conduct empirical analysis of Google's display ads system, and the effects user metrics have on display ads. 
\section{Implementation}\label{sec:ads}

We implement our influence measure to study Google's display advertising system. 
Users can set demographics (like gender or age) on the Google Ad Settings page\footnote{\texttt{google.com/settings/ads}}; these are used by the Google ad serving algorithm 
to determine which ads to serve. 
We apply our influence measure to study how demographic settings influence the targeted ads served by Google.  
We use the AdFisher tool~\cite{datta14arxiv} for automating browser activity and collect ads. 

We pick the set of features: $N = \{\mbox{gender, age, language}\}$. Feature states are $\{male,female\}$ for gender, $\{18-24,35-44,55-64\}$ for age, and $\{\mbox{English},\mbox{Spanish}\}$ for language;
this gives us $2\times 3\times 2= 12$ possible user profiles.
Using AdFisher, 
we launch twelve fresh browser instances, and assign each one a random user profile. For each browser instance, the corresponding settings are applied on the Ad Settings page, and Google ads on the BBC news page \url{bbc.com/news} are collected. For each browser, the news page is reloaded $10$ times with $5$ second intervals. 

To eliminate ads differing due to random chance, we collect ads over $100$ iterations, each comprising of $12$ browser instances, thereby obtaining data for $1200$ simulated users. In order to minimize confounding factors such as location and system specifications, all browser instances were run from the same stationary Ubuntu machine.  
The $1200$ browsers received a total of $32,451$ ads ($763$ unique); in order to reduce the amount of noise, we focus only on ads that were displayed more than 100 times, leaving a total of $55$ unique ads. 
Each user profile $\vec a$ thus has a frequency vector of all ads $v'(\vec a)\in \N^{55}$, where the $k^\th$ coordinate is the number of times ad $k$ appeared for a user profile $\vec a$. 
We normalize $v'(\vec a)$ for each ad by the total number of times that ad appeared. Thus we obtain the final value-vectors by computing ${v_k}(\vec a) = \frac{ v'_k(\vec a)}{\sum_{\vec a} v'_k(\vec a)}, \forall \vec a, \forall k \in \{1,\dots,55\}$. 

Since user profile values are vectors, we use the general distance influence measure described in Section~\ref{sec:distancemeasures3}. The pseudo-distance we use is Cosine similarity: $\cosd(\vec x, \vec y) = 1 - \frac{\vec x\cdot \vec y}{||\vec x||  \cdot ||\vec y||}$; this has been used 
Cosine similarity has been used by \cite{tschantz2014information} and \cite{Guha'10} to measure similarity between display ads.  
The influence measure for gender, age, and language were $0.124$, $0.120$, and $0.141$ respectively; in other words, no specific feature has a strong influence over ads displayed. 

We next turn to measuring feature effects on specific ads.  
Fixing an ad $k$, we define the value of a feature vector to be the number of times that ad $k$ was displayed for users with that feature vector, and use $\attvalue$ to measure influence.  

We compare the influence measures for each attribute across all the ads and identify the top ads that demonstrate high influence. The ad for which language had the highest influence ($0.167$) was a Spanish language ad, which was served only to browsers that set `Spanish' as their language on the Ad Settings page. Comparing with statistics like mean and maximum over measures across all features given in Table~\ref{table:stats}, we can see that this influence was indeed high.

\begin{table}

\centering

\begin{tabular}{l r r r}

\hline

Statistic & Gender & Age & Language \\ [0.2ex] 

\hline

Max  &			$0.07$     &   		$0.0663$  	&		$0.167$\\

Min  	& 			$0.00683$ & 		$0.00551$  	&	$0.00723$\\

Mean &			$0.0324$  	&		$0.0318$  	&		$0.0330$\\

Median &		$0.0299$  	&		$0.0310$  	&		$0.0291$\\
StdDev	&	$0.0161$  &	$0.0144$ &		$0.024$\\

\hline

\end{tabular}

\caption{Statistics over influence measures across features.}

\label{table:stats}

\end{table}


To conclude, using a general distance measure between two value-vectors, we identify that language has the highest influence on ads. By using a more fine-grained distance function, we can single out one ad which demonstrates high influence for language. While in this case the bias is acceptable, the experiment suggests that our framework is effective in pinpointing biased or discriminatory ads. 
\section{Conclusions and Future Work}\label{sec:conclusions}
In this work, we analyze influence measures for classification tasks. Our influence measure is uniquely defined by a set of natural axioms, and is easily extended to other settings. The main advantage of our approach is the minimal knowledge we have of the classification algorithm. We show the applicability of our measure by analyzing the effects of user features on Google's display ads, despite having no knowledge of Google's classification algorithm (which, we suspect, is quite complex). 

Dataset classification is a useful application of our methods; however, our work applies to extensions of TU cooperative games where agents have more than two states (e.g. OCF games~\citep{ocfgeb}). 

The measure $\attvalue$ is trivially hard to compute exactly, since it generalizes the raw Banzhaf power index, for which this task is known to be hard~\citep{compcoopbook}. That said, both the Shapley and Banzhaf values can be approximated via random sampling~\citep{bachrach2010approx}. It is straightforward to show that random sampling provides good approximations for $\attvalue$ as well, assuming a binary classifier.

Our results can be extended in several ways. The measure $\attvalue$ is the number of times a change in a feature's state causes a change in the outcome. However, a partial dataset of observations may not contain any pair of vectors $\vec a,\vec a' \in B$, such that $\vec a' = (\vec a_{-i},b)$. In Section~\ref{sec:ads}, we control the dataset, so we ensure that all feature profiles appear. However, other datasets would not be as well-behaved. Extending our influence measure to accommodate non-immediate influence is an important step towards implementing our results to other classification domains. Indeed, the next step of our work is analyzing large-scale datasets, in order to better understand the ideas behind our influence measure.  

Finally, our experimental results, while encouraging, are illustrative rather than informative: they tell us that Google's display ads algorithm is clever enough to assign Spanish ads to Spanish speakers. 
Our experimental results enumerate the number of {\em displayed ads}; this is not necessarily indicative of users' clickthrough rates. Since our users are virtual entities, we are not able to measure their clickthrough rates; a broader experiment, where user profiles correspond to actual human subjects, would provide better insights into the effects user profiling has on display advertising.

\bibliographystyle{named}
\bibliography{abb,yair,datta-privacy}
\ifnum\Appendix=1

\onecolumn

\appendix
\setcounter{secnumdepth}{1}  

\section*{Appendix:\\Influence in Classification}

\medskip

\section{Proof of Theorem~\ref{thm:monotone-weight}}
We define $\pivotal_i(b) = \{\vec x \in [0,1]^n \mid v(\vec x) = 1, v(\vec x_{-i},b) = 0\}$, to be the set of all {\em pivotal} vectors (w.r.t. $b$), and $\antipivotal_i(b) = \{\vec x \in [0,1]^n \mid v(\vec x) = 0,v(\vec x_{-i},b) = 1\}$ to be the set of all {\em anti-pivotal vectors}. We write $\Piv_i = \left\{(\vec x,b) \in [0,1]^{n+1}\mid \vec x \in \Piv_i(b)\right\}$ and $\APiv_i = \left\{(\vec x,b) \in [0,1]^{n+1}\mid \vec x \in  \APiv_i(b)\right\}$. We note that $\vol(\Piv_i) = \vol(\APiv_i)$.
Given a point $(\vec x,b) \in \Piv_i$, we know that $v(\vec x) = 0$ but $v(\vec x_{-i},b) = 1$. Therefore, the point $((\vec x_{-i},b),x_i)$ is in $\APiv_i$. We conclude that 
\begin{align*}
\attvalue_i = &	\int_{b = 0}^1|\pivotal_i(b)|+ |\antipivotal_i(b)|\partial b \\
						= & \int_{b = 0}^1\vol(\pivotal_i(b))\partial b + \int_{b = 0}^1 \vol(\antipivotal_i(b))\partial b\\
						= & \vol(\Piv_i) + \vol(\APiv_i) = 2\vol(\Piv_i)
\end{align*} 
We begin by stating a few technical lemmas. Our objective is to establish some 
volume-preserving transformations between vectors for which $j$ is pivotal, and vectors for which $i$ is pivotal. 

Thus, to show that $\attvalue_i \ge \attvalue_j$ whenever $w_i \ge w_j > 0$, it suffices to show that $\vol(\Piv_i) \ge \vol(\Piv_j)$.
\begin{lemma}\label{lem:x_i>x_j}
Suppose that $w_i > w_j > 0$; if $\vec x \in \Piv_j(b)\setminus \Piv_i(b)$ then $x_i > x_j$.
\end{lemma}
\begin{proof}
First, note that if $v(\vec x_{-j},b) = 1$ but $v(\vec x) = 0$, then $x_j < b$. 
Now, suppose that $x_i \le x_j$; we show that $(\vec x_{-j},b)\cdot \vec w \le (\vec x_{-i},b)\cdot \vec w$. Indeed, 
\begin{align*}
(\vec x_{-j},b)\cdot \vec w  \le& (\vec x_{-i},b)\cdot \vec w \iff\\
x_iw_i + bw_j \le& x_jw_j + bw_i\iff\\
x_iw_i - x_jw_j \le & b(w_i - w_j)
\end{align*}
Thus, we just need to show that $x_iw_i - x_jw_j \le b(w_i-w_j)$. Since $x_i \le x_j$, $x_i w_i - x_j w_j \le x_j(w_i - w_j)$, and since $w_i > w_j$, this is at most $b(w_i - w_j)$, as required. This means that if $x_i \le x_j$ then $\vec x \in \Piv_i(b)$, which concludes the first part of the proof.
\end{proof}
Let $f_{ij}:\R^n \to \R^n$ be the transformation
$$f_{ij}(\vec x)_k = \begin{cases}x_i & \mbox{if } k = j\\x_j & \mbox{if } k = i\\ x_k & \mbox{otherwise.}\end{cases}$$
\begin{lemma}\label{lem:f_ij}
If $\vec x \in \Piv_j(b) \setminus \Piv_i(b)$ then $f_{ij}(\vec x) \notin \Piv_j(b)\cup\APiv_j(b)$. 
\end{lemma}
\begin{proof}
First, note that $(b- x_j)(w_i - w_j) > 0$; this is because $b > x_j$ and $w_i > w_j$. This implies that $x_jw_i + bw_j < bw_i + x_jw_j$. Now, since $v(\vec x_{-i},b)  = 0$, we know that $bw_i + x_jw_j < q - \sum_{k \ne i,j}x_k w_k$; therefore, $\vec w \cdot (f_{ij}(\vec x)_{-j},b)= \sum_{k \ne i,j}x_k w_k + x_jw_i + bw_j < q$, and $v(f_{ij}(\vec x)_{-j},b) = 0$. This implies that $f_{ij}(\vec x) \notin \Piv_j(b)$. 

Now, $(x_i - x_j)(w_i - w_j) > 0$ since $x_i > x_j$ by Lemma~\ref{lem:x_i>x_j}. Therefore, $x_jw_i + x_iw_j < x_iw_i + x_jw_j < q - \sum_{k \ne i,j}x_kw_k$, which implies that $\vec w \cdot f_{ij}(\vec x) < q$, hence $v(f_{ij}(\vec x)) = 0$. In particular, $f_{ij}(\vec x) \notin \APiv_j(b)$.
\end{proof}
\begin{lemma}\label{lem:x_i> b > x_j}
Suppose $w_i > w_j > 0$ and that $\vec x \in \Piv_j(b)\setminus \Piv_i(b)$; if $f_{ij}(\vec x) \notin \Piv_i(b)$ then $x_i \ge b > x_j$ .
\end{lemma}
\begin{proof}
Suppose that $b > x_i > x_j$. We note that $(b - x_i)(w_i - w_j) > 0$, which implies that $bw_i + x_iw_j > x_i w_i + b w_j \ge q - \sum_{k \ne i,j}x_kw_k$. Hence, $\vec w\cdot (f_{ij}(\vec x)_{-i},b)> q$, which implies that $f_{ij}(\vec x) \in \Piv_i(b)$. Thus, if $f_{ij}(\vec x) \notin \Piv_i(b)$, it must be the case that $x_i \ge b > x_j$.
\end{proof}
Given some $\vec x\in [0,1]^n$ and some $b \in [0,1]$, we define $g_{ij}:[0,1]^n\times [0,1]\to[0,1]^n$ as follows:
\begin{align*}
g_{ij}(\vec x,b)_k = \begin{cases}x_j & \mbox{if } k = i\\ b & \mbox{if } k = j\\ x_k & \mbox{otherwise.}\end{cases}
\end{align*}
\begin{lemma}\label{lem:g_ij}
If $\vec x \in \Piv_j(b) \setminus \Piv_i(b)$ and $f_{ij}(\vec x) \notin \Piv_i(b)$, then $g_{ij}(\vec x,b) \in \Piv_i(x_i)\setminus (\Piv_j(x_i)\cup \APiv_j(x_i))$.
\end{lemma}
\begin{proof}
First, we observe that $(g_{ij}(\vec x,b)_{-i},x_i) = (\vec x_{-j},b)$, and that $(g_{ij}(\vec x,b)_{-j},x_i) = f_{ij}(\vec x)$. As observed in Lemma~\ref{lem:f_ij}, if $x_i > x_j$ then $v(f_{ij}(\vec x)) = 0$. Therefore, $g_{ij}(\vec x,b) \notin \Piv_j(x_j)$. Moreover, since $\vec x \in \Piv_j(b)$, $v(g_{ij}(\vec x,b)) = 1$, so $g_{ij}(\vec x,b) \in \Piv_i(x_i)$. 
On the other hand, $(b - x_j)(w_i - w_j) > 0$, so $x_jw_i + bw_j < bw_i + x_jw_j < q - \sum_{k \ne i,j}x_kw_k$, so $g_{ij}(\vec x,b)\cdot \vec w < q$. This means that $g_{ij}(\vec x,b) \notin \APiv_j(x_i)$.
\end{proof}

Given a set $S \subseteq \R^m$ and a function $f:\R^m \to \R^m$, we define $f(S) = \{f(\vec s) \mid \vec s \in S\}$. We can extend $f_{ij}$ and $g_{ij}$ defined above to functions from $\R^{n+1}$ to $\R^{n+1}$ as follows. Given a point $(\vec x,b) \in \R^{n+1}$, we define $F_{ij}(\vec x,b) = (f_{ij}(\vec x),b)$, and $G_{ij}(\vec x,b) = (g_{ij}(\vec x,b),x_i)$. We note that both $F_{ij}$ and $G_{ij}$ merely swap coordinates in their inputs, thus they preserve distances: 
$$d(G_{ij}(\vec x,b), G_{ij}(\vec y,c)) = d((\vec x,b),(\vec y,c))$$ 
for any metric $d$. Isoperimetric transformations are known to preserve volume: if $I:\R^m \to \R^m$ is an isoperimetry, then $\vol(S) = \vol(I(S))$ for any $S \subseteq \R^m$.

\begin{theorem}\label{thm:piv}
If $w_i \ge w_j>  0$ then $\vol(\Piv_j) \le \vol(\Piv_i)$.
\end{theorem}
\begin{proof}

We partition $\Piv_j$ as follows. We denote 
\begin{align*}
A_{ij} &= \Piv_j\cap \Piv_i,\\
B_{ij} &= \left\{(\vec x,b) \in \Piv_j \setminus \Piv_i\mid (f_{ij}(\vec x),b) \in \Piv_i\right\}, \mbox{ and}\\
C_{ij} &= \left\{(\vec x,b)\in \Piv_j\setminus \Piv_i \mid (f_{ij}(\vec x),b) \notin\Piv_i\right\}.
\end{align*}
Clearly, $A_{ij},B_{ij}$ and $C_{ij}$ partition $\Piv_j$. 

According to Lemma~\ref{lem:f_ij}, $F_{ij}(B_{ij}) \subseteq \Piv_i\setminus \Piv_j$. Now, let us observe $C_{ij}$. According to Lemma~\ref{lem:g_ij}, $G_{ij}(C_{ij}) \subseteq \Piv_i\setminus \Piv_j$. 
It remains to show that $F_{ij}(B_{ij}) \cap G_{ij}(C_{ij}) = \emptyset$. 
Suppose that there are some $(\vec x,b) \in B_{ij},(\vec z,c) \in C_{ij}$ such that $(f_{ij}(\vec x),b) = (g_{ij}(\vec z,c),z_i)$. This means that $(\vec z,c) = ((\vec x_{-i},b),x_i)$. 
To prove a contradiction, it suffices to show that if $(\vec x,b) \in B_{ij}$ then we have that $((\vec x_{-i},b),x_i) \notin C_{ij}$. 
In order to be in $C_{ij}$, it must be the case that $f_{ij}(\vec x_{-i},b) \notin \Piv_i(x_i)$; we show that $f_{ij}(\vec x_{-i},b) \in \Piv_i(x_i)$. First, let us write $f_{ij}(\vec x_{-i},b) = \vec y$. We note that $y_k = x_k$ for all $k \ne i,j$, that $y_j = b$, and that $y_i = x_j$. Since $b > x_j$, it must be the case that $(b - x_j)(w_i - w_j) > 0$, hence $bw_i + x_jw_j > x_j w_i + bw_j$. Therefore, $\vec w\cdot \vec y < \vec w\cdot (\vec x_{-i},b)$. Now, since $(\vec x,b) \in \Piv_j\setminus \Piv_i$, it must be the case that $v(\vec x_{-i},b) = 0$, i.e. that $\vec w\cdot (\vec x_{-i},b) < q$. This means that $v(\vec y) = 0$. 
We now show that $v(\vec y_{-i},x_i) = 1$. Since $y_i = x_j$ and $y_j = b$, $(\vec y_{-i},x_i) = (\vec x_{-j},b)$. Since $(\vec x,b) \in \Piv_j$, $v(\vec y_{-i},x_i) = v(\vec x_{-j},b) = 1$. Therefore, $\vec y \in \Piv_i$, and thus $((\vec x_{-i},b),x_i)\notin C_{ij}$. We conclude that indeed $F_{ij}(B_{ij}) \cap G_{ij}(C_{ij}) = \emptyset$.

To conclude, 
\begin{align*}
\vol(\Piv_j) 	= & \vol(A_{ij}) + \vol(B_{ij}) + \vol(C_{ij})\\
							= & \vol(A_{ij}) + \vol(F_{ij}(B_{ij})) + \vol(G_{ij}(C_{ij}))\\ \le &\vol(\Piv_i)
\end{align*}
which concludes the proof. 
\end{proof}
\begin{corollary}\label{cor:monotone-weight}
Let $\cal G = \tup{N,[0,1]^n,v}$ be a game where $v$ is a linear separator given by $\vec w$ and $q$. If $w_i \ge w_j > 0$ then $\attvalue_i(\cal G) \ge \attvalue_j(\cal G)$.
\end{corollary}
Corollary~\ref{cor:monotone-weight} shows that $\attvalue$ is monotone in feature weights. a complementary result shows that increasing a feature's weight would result in an increase in influence. 
Next, we show that Corollary~\ref{cor:monotone-weight} holds even when weights are negative. 
\begin{lemma}\label{lem:2player-monotoneweight}
Let $\cal G = \tup{\{1,2\},[0,1]^2,v}$ be a 2-feature linear separator with $w_1 \ge 0$ and $w_2 < 0$. Then $\attvalue_1 (\cal G) > \attvalue_2(\cal G)$ if and only if $|w_1| > |w_2|$.
\end{lemma}
\begin{proof}
We begin by assuming that $q \ge 0$.
First, suppose that $w_1 < q$. In that case, for all $(x_1,x_2) \in [0,1]^2$, we have $x_1w_1 + x_2w_2 \le x_1w_1 \le w_1 < q$, so $v(x_1,x_2) = 0$ for all $(x_1,x_2) \in [0,1]^2$. In particular, $\attvalue_1(\cal G) = \attvalue_2(\cal G) = 0$ and we are done. 

We now assume that $w_1 \ge q$. We show that the claim holds by direct computation of $\attvalue_1,\attvalue_2$. 
We start by computing $\attvalue_1(\cal G)$. By definition, $\attvalue_1(\cal G)$ equals 
\begin{align*}
\int_0^1 \left(\int_0^1 \I(v(x_1,x_2) = 1)\partial x_1 \int_0^1 \I(v(y_1,x_2) = 0)\partial y_1\right) \partial x_2
\end{align*}
which equals
\begin{align}
\int_0^1 \left(\int_0^1 \I(x_1\ge \frac{q - x_2w_2}{w_1})\partial x_1 \int_0^1 \I(y_1 < \frac{q - x_2w_2}{w_1})\partial y_1 \right)\partial x_2 \label{eq:attvalue1}
\end{align}
The internal integrals in~\eqref{eq:attvalue1} are zero whenever $\frac{q - x_2w_2}{w_1}\notin [0,1]$. We know that $\frac{q - x_2w_2}{w_1}\ge 0$ for all $x_2 \in [0,1]$; however, $\frac{q - x_2w_2}{w_1}\le 1$ only when $x_2 \le \frac{w_1 - q}{-w_2}$. This inequality is non trivial only if $\frac{w_1 - q}{-w_2} \le 1$. This happens only when $q \ge w_1 + w_2$.
Therefore, we distinguish between two cases; the first case is when $q \ge w_1 + w_2$, and the second is when $q < w_1 + w_2$. In the second case, since $q > 0$, $w_1 + w_2 > 0$ as well, hence $|w_1| > |w_2|$. In the first case we have:
\begin{align}
\attvalue_1(\cal G) = & \int_0^{\frac{w_1 - q}{-w_2}} \left(1 - \frac{q - x_2w_2}{w_1}\right)\left(\frac{q - x_2w_2}{w_1}\right)\partial x_2 \nonumber\\
 = & \frac{(w_1 - q)^2(2q+w_1)}{6(-w_2)w_1^2}\label{eq:attvalue1case1}
\end{align}
In the second case we have 
\begin{align}
\attvalue_1(\cal G) = & \int_0^{1} \left(1 - \frac{q - x_2w_2}{w_1}\right)\left(\frac{q - x_2w_2}{w_1}\right)\partial x_2 \nonumber\\
 = &  \frac{6q(w_1 + w_2) - 6q^2 - w_2(3w_1+2w_2)}{6w_1^2}\label{eq:attvalue1case2}
\end{align}
Now, let us proceed to compute $\attvalue_2(\cal G)$. We have that $\attvalue_2(\cal G)$ equals
\begin{align*}
\int_0^1\left(\int_0^1 \I(v(x_1,x_2) = 1)\partial x_2\int_0^1 \I(v(x_1,y_2) = 0)\partial y_2\right)\partial x_1\nonumber
\end{align*}
which equals
\begin{align}
\int_0^1\left(\int_0^1 \I(x_2\le \frac{x_1w_1 - q}{-w_2})\partial x_2\int_0^1 \I(y_2 >\frac{x_1w_1 - q}{-w_2})\partial y_2\right)\partial x_1\label{eq:attvalue2}
\end{align}
Again, the internal integrals in~\eqref{eq:attvalue2} are not zero only if $\frac{x_1w_1 - q}{-w_2} \in [0,1]$. $\frac{x_1w_1 - q}{-w_2} \ge 0$ if and only if $x_1\ge \frac{q}{w_1}$, and $\frac{x_1w_1 - q}{-w_2} \le 1$ if and only if $x_1 \le \frac{q - w_2}{w_1}$. This inequality is non-trivial only if $\frac{q - w_2}{w_1} < 1$, which happens only when $q < w_1+w_2$. Thus, we again distinguish between the case when $q \ge w_1 + w_2$ and the case when $q < w_1 + w_2$. In the first case, we have 
\begin{align}
\attvalue_2(\cal G) = & \int_{\frac{q}{w_1}}^1\left(\frac{x_1w_1 - q}{-w_2}\right)\left(1 - \frac{x_1w_1 - q}{-w_2}\right)\partial x_2\nonumber\\
 =& \frac{(w_1 - q)^2(2q - 2w_1 - 3w_2)}{6w_2^2w_1}\label{eq:attvalue2case1}
\end{align}
and in the second case, $\attvalue_2(\cal G)$ equals
\begin{align}
\int_{\frac{q}{w_1}}^{\frac{q - w_2}{w_1}}\left(\frac{x_1w_1 - q}{-w_2}\right)\left(1 - \frac{x_1w_1 - q}{-w_2}\right)\partial x_2 =& \frac{-w_2}{6w_1}\label{eq:attvalue2case2}
\end{align}
Let us compare the values when $q \ge w_1+w_2$.
\begin{align}
\attvalue_1(\cal G) \ge & \attvalue_2(\cal G) & \iff\nonumber\\
\frac{(w_1 - q)^2(2q+w_1)}{6(-w_2)w_1^2} \ge &\frac{(w_1 - q)^2(2q - 2w_1 - 3w_2)}{6w_2^2w_1}& \iff\nonumber\\
\frac{2q+w_1}{w_1} \ge &\frac{2q - 2w_1 - 3w_2}{-w_2}& \iff\nonumber\\
(-w_2)(2q+w_1) \ge &w_1(2q - 2w_1 - 3w_2)& \iff\nonumber\\
w_1(w_1+w_2) \ge& q(w_1+w_2)&\label{eq:ineq-case1}
\end{align}
Thus,~\eqref{eq:ineq-case1} holds with equality if $w_1 = -w_2$, $\attvalue_1(\cal G) > \attvalue_2(\cal G)$ if $w_1 > -w_2$ (since $w_1 > q \ge 0$ by assumption), and $\attvalue_1(\cal G) < \attvalue_2(\cal G)$ otherwise. For the second case, we have 
\begin{align}
\attvalue_1(\cal G) \ge & \attvalue_2(\cal G) & \iff\nonumber\\
\frac{6q(w_1 + w_2) - 6q^2 - w_2(3w_1+2w_2)}{6w_1^2} \ge & \frac{-w_2}{6w_1} & \iff\nonumber\\
\frac{6q(w_1 + w_2) - 6q^2 - w_2(3w_1+2w_2)}{w_1}\ge & -w_2 & \iff\nonumber\\
6q(w_1 + w_2) - 6q^2 - w_2(3w_1+2w_2)\ge & (-w_2)w_1 & \iff\nonumber\\
6q(w_1 + w_2) - 6q^2 - 2w_2(w_1+w_2)\ge & 0 & \iff\nonumber\\
(3q - w_2)(w_1 + w_2) \ge & 3q^2 & \label{eq:ineq-case2}
\end{align}
Now,~\eqref{eq:ineq-case2} holds with equality if $w_1 + w_2 = 0$, since then $q = 0$ as well. Finally, if $w_1 + w_2 > 0$, then it holds with strict inequality since $w_1 + w_2 \ge q$ and $3q - w_2 > 3q$, and we are done.

Next, let us assume that $q < 0$. 
We again directly compute $\attvalue_1(\cal G)$ and $\attvalue_2(\cal G)$. First, if $w_2 \ge q$, then $x_1w_1 + x_2 w_2 \ge x_2 w_2 \ge w_2 \ge q$ for all $(x_1,x_2) \in [0,1]^2$; hence $\attvalue_1(\cal G) = \attvalue_2(\cal G) = 0$, and the claim trivially holds. We now assume that $w_2 < q$. Again, we have that $\attvalue_1(\cal G)$ equals
\begin{align}
\int_0^1 \left(\int_0^1\I(x_1 \le \frac{q - x_2w_2}{w_1})\partial x_1 \int_0^1 \I(y_1 > \frac{q - x_2w_2}{w_1})\partial y_1\right)\partial x_2
\label{eq:attvalue1-neg}
\end{align}
We need to have $\frac{q - x_2w_2}{w_1} \in [0,1]$. $\frac{q - x_2w_2}{w_1} \ge 0$ if and only if $x_2 \ge \frac{q}{w_2}$. Since $w_2 < q$, this value is always less than 1. Moreover, $\frac{q - x_2w_2}{w_1}  \le 1$ if and only if $x_2 \le \frac{q - w_1}{w_2}$. This inequality is not trivial only if $\frac{q - w_1}{w_2} \le 1$, which happens whenever $q \ge w_2 + w_1$. Thus, when $q \ge w_1+w_2$, $\attvalue_1(\cal G)$ equals 
\begin{align*}
\int_{\frac{q}{w_2}}^{\frac{q - w_1}{w_2}} \left(\frac{q - x_2w_2}{w_1}\right)\left(1 - \frac{q - x_2w_2}{w_1}\right)\partial x_2 & = \frac{w_1}{-6w_2}
\end{align*}
and when $q < w_1 + w_2$, $\attvalue_1(\cal G)$ equals 
\begin{align*}
\int_{\frac{q}{w_2}}^1 \left(\frac{q - x_2w_2}{w_1}\right)\left(1 - \frac{q - x_2w_2}{w_1}\right)\partial x_2 & = \nonumber\\
\frac{(q - w_2)^2(2q - 2w_2 - 3w_1)}{6w_2w_1^2}
\end{align*}
For $\attvalue_2(\cal G)$, we employ a similar reasoning. First, $\attvalue_2(\cal G)$ equals
\begin{align}
\int_0^1 \left(\int_0^1\I(x_2 \le \frac{x_1w_1 - q}{-w_2})\partial x_2 \int_0^1 \I(y_2 > \frac{x_1w_1 - q}{-w_2})\partial y_2 \right)\partial x_1\label{eq:attvalue2-neg}
\end{align}
And again, $\frac{x_1w_1 - q}{-w_2} \in [0,1]$ if and only if $x_1 \le \frac{q - w_2}{w_1}$. Note that since $w_2 < q$, $\frac{q - w_2}{w_1} \ge 0$. This constraint is only meaningful when $q < w_1 + w_2$. Thus, when $q \ge w_1+w_2$, we have that $\attvalue_2(\cal G)$ equals
\begin{align*}
\int_0^1 \left(\frac{x_1w_1 - q}{-w_2}\right)\left(1 - \frac{x_1w_1 - q}{-w_2}\right)\partial x_1 & = \nonumber\\
 - \frac{6q^2 - 6q(w_1+w_2) + w_1(3w_2 + 2w_1)}{6w_2^2}
\end{align*}
and equals 
$$  - \frac{(q - w_2)^2(2q+w_2)}{6w_2^2w_1}$$
otherwise. 

Next, we compare the values we obtained. When $q \ge w_1 + w_2$, we have that $w_1 + w_2 < 0$, and in particular, $|w_2| > |w_1|$. Moreover,
\begin{align*}
- \frac{6q^2 - 6q(w_1+w_2) + w_1(3w_2 + 2w_1)}{6w_2^2}\ge  & \frac{w_1}{-6w_2} &\iff\\
\frac{-6q^2 + 6q(w_1+w_2) - w_1(3w_2 + 2w_1)}{-w_2}\ge & w_1 &\iff\\
-6q^2 + 6q(w_1+w_2) - w_1(2w_2 + 2w_1)\ge & 0&\iff\\
(3q- w_1)(w_2 + w_1)\ge & 3q^2
\end{align*}
Under our assumptions, this inequality holds, and we are done with the first case. For the second case, 
\begin{align*}
\frac{(q - w_2)^2(2q - 2w_2 - 3w_1)}{6w_2w_1^2} \ge & - \frac{(q - w_2)^2(2q+w_2)}{6w_2^2w_1} & \iff\\
\frac{2w_2 +3w_1- 2q}{w_1} \ge & - \frac{2q+w_2}{-w_2} & \iff\\
(-w_2)(2w_2 +3w_1- 2q) \ge & w_1(-2q-w_2) & \iff\\
(-w_2)(w_1+w_2) \ge& (-q)(w_1+w_2)&
\end{align*}
Since $w_2< q$, this inequality holds with equality when $w_1 = -w_2$, it is strict whenever $|w_1| > |w_2|$, and the reverse holds when $|w_1| < |w_2|$.
\end{proof}
We are now ready to complete the proof of Theorem~\ref{thm:monotone-weight}.
\begin{proof}[Proof of Theorem~\ref{thm:monotone-weight}]
We have shown the case where $w_i \ge w_j \ge 0$ in Theorem~\ref{thm:piv}. We have also shown this to be true for two features in Lemma~\ref{lem:2player-monotoneweight}. We just need to show that Lemma~\ref{lem:2player-monotoneweight} extends to the case of arbitrary players. Suppose that $|w_i| > |w_j|$. Let us write $\attvalue_i(\tup{N,\vec w;q})$ to be the influence of $i$ under the linear classifier defined by $\vec w$ and $q$. We observe that 
\begin{align*}
\attvalue_i(\tup{N,\vec w;q}) = &\int_{\vec x_{-i,-j}}\attvalue_i(\tup{\{i,j\},(w_i,w_j);q - \sum_{k \ne i,j} x_kw_k})\\
\ge & \int_{\vec x_{-i,-j}}\attvalue_j(\tup{\{i,j\},(w_i,w_j),q - \sum_{k \ne i,j} x_kw_k}) \\ = &\attvalue_j(\tup{N,\vec w;q})
\end{align*}
which concludes the proof.
\end{proof}

\section{Proof that $\attvalue$ satisfies (D), (Sym) and (DU)}
We show that $\attvalue$ satisfies the three axioms. If $v(\vec a_{-i},b) = v(\vec a)$ for all $\vec a \in A$ and all $b \in A_i$, then $|v(\vec a_{-i},b) - v(\vec a)| = 0$, and in particular, $\attvalue_i(\cal G) = 0$; hence, $\attvalue$ satisfies the dummy property. Suppose we are given a bijection $\sigma_i: A_i \to A_i$. We observe that 
\begin{align*}
\attvalue_i(\cal G) =&\frac{1}{|A|}\sum_{\vec a \in A}\sum_{b \in A_i} |v(\vec a_{-i},b) - v(\vec a)|\\
			= & \frac{1}{|A|}\sum_{\vec a_{-i} \in A_{-i}}\sum_{b'\in A_i}\sum_{b \in A_i} |v(\vec a_{-i},\sigma_i(b)) - v(\vec a_{-i},\sigma_i(b'))|\\
			=& \frac{1}{|A|}\sum_{\vec a_{-i} \in A_{-i}}\sum_{b'\in A_i}\sum_{b \in A_i} |v_{\sigma_{i}}(\vec a_{-i},b) - v_{\sigma_i}(\vec a_{-i},b')|\\
			=& \frac{1}{|A|}\sum_{\vec a \in A}\sum_{b'\in A_i}\sum_{b \in A_i} |v_{\sigma_{i}}(\vec a_{-i},b) - v_{\sigma_i}(\vec a)| = \attvalue_i(\sigma_i\cal G)
\end{align*}
so $\attvalue$ is invariant under permutations of feature states. Similarly, for any bijection $\sigma:N \to N$, $\attvalue_i(\cal G) = \attvalue_{\sigma(i)}(\sigma\cal G)$; therefore, $\attvalue$ satisfies symmetry. 

Given a set $B \subseteq A$ and a feature $i$, let us write $W_{\bar{\vec a}_{-i}}(B)= \{\vec a \in B\mid v(\vec a) = 1,\vec a_{-i} = \bar{\vec a}_{-i}\}$, and $L_{\bar{\vec a}_{-i}}(B) = \{\vec a \in B\mid v(\vec a) = 0,\vec a_{-i} = \bar{\vec a}_{-i}\}$. We observe that $W_{\vec a_{-i}}(B)\cap W_{\bar{\vec a}_{-i}}(B) = L_{\vec a_{-i}}(B)\cap L_{\bar{\vec a}_{-i}}(B) = \emptyset$; moreover, $L(B) = \bigcup_{\vec a_{-i} \in A_{-i}} L_{\vec a_{-i}}(B)$ and $W(B) = \bigcup_{\vec a_{-i} \in A_{-i}} W_{\vec a_{-i}}(B)$.
Now, given some $B \subseteq A$, let us take some $W'\subseteq W(A) \setminus W(B)$. 
\begin{align*}
\attvalue_i(W(B),L(B))=& \sum_{\vec a \in B}\sum_{\substack{b \in A_i: \\(\vec a_{-i},b) \in B}}|v(\vec a_{-i},b) - v(\vec a)|\\
											=& \sum_{\vec a \in W(B)}\sum_{\substack{b \in A_i:\\(\vec a_{-i},b) \in B}}|v(\vec a_{-i},b) - v(\vec a)|\\
											 &+\sum_{\vec a \in L(B)}\sum_{\substack{b \in A_i:\\(\vec a_{-i},b) \in B}}|v(\vec a_{-i},b) - v(\vec a)|
\end{align*}
Next, we observe that the first summand equals
\begin{align*}
\sum_{\vec a \in W(B)}\sum_{\substack{b \in A_i:\\(\vec a_{-i},b) \in B}}v(\vec a) - v(\vec a_{-i},b), 
\end{align*}
which equals
\begin{align}\label{eq:winloss1}
\sum_{\vec a_{-i} \in A}\sum_{\vec a \in W_{\vec a_{-i}}(B)}\sum_{\substack{b \in A_i:\\(\vec a_{-i},b) \in B}}v(\vec a) - v(\vec a_{-i},b) 
\end{align}
Now, $v(\vec a) - v(\vec a_{-i},b) = 1$ if and only if $v(\vec a_{-i},b) = 0$; that is, if $(\vec a_{-i},b) \in L_{\vec a_{-i}}(B)$. Thus, Equation~\eqref{eq:winloss1} equals  
\begin{align}\label{eq:winloss2}
\sum_{\vec a_{-i} \in A}\sum_{\vec a \in W_{\vec a_{-i}}(B)}|L_{\vec a_{-i}}(B)| & =\\
\sum_{\vec a_{-i} \in A}|W_{\vec a_{-i}}(B)||L_{\vec a_{-i}}(B)| \nonumber
\end{align}
A similar construction with $W'$ shows that 
$$\attvalue_i(W',L(B)) = \sum_{\vec a_{-i} \in A_{-i}}|W_{\vec a_{-i}}'|\cdot |L_{\vec a_{-i}}(B)|;$$ 
since $W(B)$ and $W'$ are disjoint, $\attvalue$ satisfies the disjoint union property. 

\section{Relation to Classic Values in TU Cooperative Games}
Our work generalizes influence measurement in classic TU cooperative games. We recall that a cooperative game with transferrable utility is given by a set of players $N =\{1,\dots,n\}$, and a function $v:2^N \to \R$, called the {\em characteristic function}. A game is defined by the tuple $\cal G = \tup{N,v}$. 
We say that a game $\cal G$ is {\em monotone} if for all $S \subseteq T \subseteq N$, $v(S) \le v(T)$. 

Classic literature identifies two canonical methods of measuring feature influence in cooperative games, the Shapley value~\citep{shapleyvalue}, and the Banzhaf value~\citep{banzhaf}. 
We begin by providing the following definitions. 
Given a set $S \subseteq N$ and a player $i$, we let $m_i(S) = v(S\cup\{i\}) - v(S)$ denote the {\em marginal contribution} of $i$ to $S$. The value $m_i(S)$ simply describes the added benefit of having $i$ join the coalition $S$.
Let $\Pi(N)$ be the set of all bijections from $N$ to itself (also called the set of permutations of $N$); given some $\sigma \in \Pi(N)$ We let $P_i(\sigma) = \{j \in N\mid \sigma(j) < \sigma(i)\}$ be the set of the predecessors of $i$ under $\sigma$. 
We define $m_i(\sigma) = v(P_i(\sigma) \cup\{i\}) - v(P_i(\sigma))$. 
\begin{definition}
The {\em Banzhaf value} of a player $i \in N$ is given by 
$$\banzhaf_i(\cal G) = \frac{1}{2^{n}}\sum_{S \subseteq N}m_i(S).$$
\end{definition}
The Banzhaf value takes on a simple probabilistic interpretation: if we choose a set $S$ uniformly at random from $N$, the Banzhaf value of a player is his expected marginal contribution to that set. 

Rather than uniformly sampling sets, the Shapley value is based on uniformly sampling permutations. 
\begin{definition}
The {\em Shapley value} of a player $i \in N$ is given by
$$\sv_i(\cal G) = \frac{1}{n!}\sum_{\sigma \in \Pi(N)} v(P_i(\sigma)\cup\{i\}) - v(P_i(\sigma)).$$
\end{definition}

Intuitively, one can think of the Shapley value as the result of the following process. We randomly pick some order of the players; each player receives a payoff that is equal to his marginal contribution to his predecessors in the ordering. The Shapley value is simply the expected payoff a player receives in this scheme. 

When we sample sets uniformly at random from $N \setminus \{i\}$, we are heavily biased towards selecting sets whose size is approximately $n/2$. When measuring influence according to the Shapley value, we are no longer biased towards any set size. One can think of the Shapley value is measuring a player's expected marginal contribution to a set $S$, where $S$ is chosen according to the following process. First, we pick some $k \in\{0,\dots,n-1\}$ uniformly at random, and then we pick a set of size $k$ uniformly at random. 

We observe that our classification setting is a generalization of TU cooperative games. Think of each player as a feature that can take on two values: 0 (corresponding to ``absent''), and 1 (corresponding to ``present'').
An immediate observation is that $\stateval$ coincides with the Banzhaf value for TU cooperative games. Is there some natural extension of the Shapley value for general classification tasks?

Our work provides a negative answer to this question. We observe that Theorem~\ref{thm:charstateval} states that the only value that satisfies the dummy, symmetry and linearity axioms is $\stateval$. When reduced to the cooperative game setting, we obtain axioms that were used to axiomatically characterize both the Shapley and the Banzhaf values~\citep{lehrer1988axiomatization,shapleyvalue,shapleyyoung}. 

The dummy axiom (Definition~\ref{def:dummy}) reduces to the following: a player $i \in N$ is a dummy if for all $S \subseteq N$, $v(S\cup\{i\}) = v(S)$. Thus, the dummy axiom requires that if a player is a dummy, then his value should be zero. 

The symmetry axiom (Definition~\ref{def:symmetry}) reduces to the following: 
given a game $\cal G =\tup{N,v}$, and some $i,j\in N$, let us define $\cal G' = \tup{N,v'}$ as follows: for all $S \subseteq N\setminus\{i,j\}$, $v'(S) = v(S)$, and $v'(S\cup\{i,j\}) = v(S\cup\{i,j\})$; however, $v'(S\cup\{i\}) = v(S\cup\{j\})$ and $v'(S\cup\{j\}) = v(S\cup\{i\})$. A value $\phi$ satisfies symmetry if 
$\phi_i(\cal G) = \phi_j(\cal G')$. 
Symmetry reduces to saying that if we replace $v(S)$ with $v(S\setminus i)$ for all $S$ such that $i \in S$, and replace $v(S)$ with $S \cup\{i\})$ for all $S$ such that $i \notin S$, then the total influence of a player (i.e. his influence when being absent plus his influence when present) does not change. 

Additivity as defined in Definition~\ref{def:additivity} is also naturally applied to TU cooperative games and is equivalent to the definition given in other axiomatic treatments of values in cooperative games. 

It is well-known that both the Banzhaf and Shapley values satisfy the dummy, symmetry and additivity axioms, and indeed, Proposition~\ref{prop:impossibility1} applies to them both: the Banzhaf value (and Shapley) of a player only measures the effect of player $i$ joining a coalition, but not the effect of him leaving it. These two values, however, sum to 0. Indeed:
\begin{align*}
\banzhaf_{i,1}(\cal G) + \banzhaf_{i,0}(\cal G) = & \frac{1}{2^n}\sum_{S \subseteq N}v(S\cup\{i\} )- v(S) \\
&+ \frac{1}{2^n}\sum_{S \subseteq N}v(S\setminus \{i\}) - v(S) \\
= & \frac{1}{2^n}\sum_{S \subseteq N\setminus \{i\}} v(S \cup \{i\})- v(S) \\
&+ \frac{1}{2^n}\sum_{S \subseteq N \setminus\{i\}}v(S) - v(S\cup\{i\})
\\ 
=& 0
\end{align*}
Theorem~\ref{thm:attvalue} characterizes $\attvalue$ as the unique value to satisfy the dummy, symmetry and disjoint union properties. 

Going back to the classification setting, it is easy to see that Definition~\ref{def:disjoint-union} implies that for $C \subseteq A$ and any two sets $B,B' \subseteq A\setminus C$, $\phi_i(B,C) + \phi_i(B',C) = \phi_i(B\cup B',C) + phi_i(B\cap B',C)$. 

One can directly interpret the DU property in TU cooperative games. Given a game $\cal G = \tup{N,v}$ and a subset $\cal B$ of $2^N$, both the Shapley and Banzhaf values can be defined to ignore any elements that are not contained in $\cal B$. It is easy to see that Theorem~\ref{thm:attvalue} implies the uniqueness of $\attvalue$ for TU cooperative games, and that it equals the Banzhaf value. Thus, Theorem~\ref{thm:attvalue} can be seen as an alternative axiomatization of the Banzhaf value, this time from the binary classification perspective. 

\section{Axiomatic Approach to State Influence}\label{sec:stateinfluence}
Section~\ref{sec:featureinfluence} provided an answer to questions of the following form: what is the impact of gender on classification. The answer provided in previous sections was that influence was a function of the feature's ability to change outcomes by changing its state. 

It is also useful to ask a related question: suppose that a certain search engine user is profiled as a female. What is the influence of this profiling decision? In other words, rather than measuring feature influence, we are measuring the influence of feature $i$ being in a certain {\em state}.

For a feature $i \in N$ and a state $b \in N$, we can ask what is the influence of the state $b$, rather than the influence of $i$. That is, rather than having a value $\phi_i(\cal G)$ for a feature $i \in N$, we now study the influence of the state $b \in A_i$, i.e. a real value $\phi_{i,b}(\cal G)$ for each $i \in N$ and $b \in A_i$.

While Proposition~\ref{prop:impossibility1} implies that any {\em feature} influence measure that satisfies the dummy, symmetry and additivity axioms must be trivial, this result does not carry through to measures of state influence. 
\begin{description}
\item[Dummy (D):] given $i \in N$ and $b \in A_i$, we say that $\alpha$ satisifies the dummy property if whenever $v(\vec a_{-i},b) = v(\vec a)$ for all $\vec a \in A$, $\alpha_{i,b} = 0$.
\item[Symmetry (Sym):] Two states $b,b' \in A_i$ are symmetric if for all $\vec a \in A$, $v(\vec a_{-i},b) = v(\vec a_{-i},b')$. A value $\alpha$ satisfies symmetry if $\alpha_{i,b} = \alpha_{i,b'}$ whenever $b$ and $b'$ are symmetric.
\item[Linearity (L):] Given games $\cal G_1 = \tup{N,A,v_1}$ and $\cal G_2 = \tup{N,A,v_2}$, let us write $\cal G = \tup{N,A,v}$ where $v = v_1+v_2$. We assume that $v_1$ and $v_2$ are such that $v$ is still a function with binary values (i.e. if $v_1(\vec a) = 1$ then $v_2(\vec a) = 0$). A value $\alpha$ is linear if $\alpha_{i,b}(\cal G) = \alpha_{i,b}(\cal G_1) + \alpha_{i,b}(\cal G_2)$.
\end{description}
Let us define 
\begin{align}\label{eq:stateval}
\stateval_{i,b}(\cal G) = \frac{1}{|A|}\sum_{\vec a \in A}v(\vec a_{-i},b) - v(\vec a)
\end{align}
We let $\raw\stateval$ denote the value $\stateval$ without the normalizing factor $\frac{1}{|A|}$. We refer to $\raw\stateval$ as the {\em raw} version of $\stateval$. In Theorem~\ref{thm:charstateval}, we show that $\raw\stateval$ is the unique (up to a constant) value that satisfies the symmetry, dummy and linearity axioms.

\begin{theorem}\label{thm:charstateval}
If a value $\phi$ satisfies the (D), (Sym), and (L), then $\phi = c\stateval$, where $c$ is an arbitrary constant. 
\end{theorem}
\begin{proof}
Let us observe that every game $v:A \to \{0,1\}$ can be written as the disjoint sum of unanimity games; namely $v = \sum_{\vec a \in A: v(\vec a) = 1} u_{\vec a}$. Thus, it suffices to show that the claim holds for unanimity games.

Let $\cal U_{\vec a}= \tup{N,A,u_{\vec a}}$; we show that $\phi_{i,b}(\cal U_{\vec a})$ equals $\stateval_{i,b}(\cal U_{\vec a})$. 
First, if $b = a_i$ then $\raw{\stateval}_{i,b}(\cal U_{\vec a}) = |A_i| - 1$; if $b \ne a_i$, then $\raw{\stateval}_{i,b}(\cal U_{\vec a}) = -1$. Now, by symmetry, we have that $\phi_{i,b}(\cal U_{\vec a})  =\phi_{i,b'}(\cal U_{\vec a})$ for all $b,b' \ne a_i$. If we write $\phi_{i,b}(\cal U_{\vec a}) = y$ for all $b \ne a_i$, and $\phi_{i,a_i}(\cal U_{\vec a}) = x$, then according to Proposition~\ref{prop:impossibility1}, $\sum_{b \ne a_i} y + x = 0$, which implies that $x =  - y(|A_i|-1)$. Finally, according to feature symmetry, the value of $y$ cannot depend on $i$, and is equal for all $j \in N$. We conclude that for all $i \in N$ and all $b \in A_i$, $\phi_{i,b}(\cal G) = \stateval_{i,b}(\cal G)$.
\end{proof}

As a direct corollary of Theorem~\ref{thm:piv}, we have that the unique (up to a constant) state value to satisfy (Sym), (D) and (DU) axioms (see Definitions~\ref{def:symmetry},~\ref{def:dummy} and~\ref{def:disjoint-union} in Section~\ref{sec:featureinfluence}) is 
$$\attvalue_{i,b}(\cal G) = \sum_{\vec a\in A}|v(\vec a_{-i},b) - v(\vec a)|.$$

\section{Influence in Weighted Settings}\label{sec:weighted}
Unlike previous sections, let us assume that there is some weight function $w:A \to \R$ that assigns a non-negative weight to every state vector. $w$ can be thought of as a prior distribution that governs the likelihood of observing a state vector $\vec a \in A$. Given $B \subseteq A$, let $w(B)$ denote $\sum_{\vec a \in B} w(\vec a)$. We also write for a given $b \in A_i$, $w(b\mid \vec a_{-i}) = \sum_{\vec a_{-i}\in A_{-i}} w(\vec a_{-i},b)$; for a given $\vec a_{-i} \in A_{-i}$, we write $w(\vec a_{-i}) = \sum_{b \in A_i} w(\vec a_{-i},b)$. 
Given this definition, let us rethink the disjoint union property. Given a set of winning state vectors $W \subseteq A$ and a set of losing state vectors $L \subseteq A$, we can think of a weighted influence measure as a function $\phi_i$ of $W,L$ and $w:A \to \R_+$.

Fix some $C \subseteq A$. Given two functions $w,w':A \to \R_+$ that agree on $C$ (i.e. $w(\vec a) = w'(\vec a)$ for all $\vec a \in C$), and some $B \subseteq A\setminus C$, let us write 
$$w\oplus_B w'(\vec a) =\begin{cases}w(\vec a) & \mbox{if }\vec a \in C\\ w(\vec a) + w'(\vec a) & \mbox{if } \vec a \in B.\end{cases} $$
\begin{definition}\label{def:disjoint-union-weighted}
We say that an influence measure satisfies {\em weighted disjoint union} (WDU) if for any disjoint $B,C \subseteq A$ and any two weight functions $w,w':A\to \R_+$ that agree on $C$, we have that $\phi_i(B,C,w) + \phi_i(B,C,w') = \phi_i(B,C,w\oplus_B w')$.
\end{definition}
\begin{lemma}
Weighted disjoint union implies the disjoint union property.
\end{lemma}
We again write $W_{\vec a_{-i}} = \{(\vec a_{-i},b) \in A \mid v(\vec a_{-i},b) =1\}$, and $L_{\vec a_{-i}} = \{(\vec a_{-i},b)\in A \mid v(\vec a_{-i},b) = 0\}$.

Given a weight function $w:A \to \R_+$ and a game $\cal G =\tup{N,A,v}$, let
$$\pattvalue(\cal G,w) = \sum_{\vec a \in A}w(\vec a)\sum_{b \in A_i}w(b|\vec a_{-i})|v(\vec a_{-i},b) - v(\vec a)|.$$
Let us extend the symmetry axiom (Definition~\ref{def:symmetry}) to a weighted variant. 
Given a weight function $w:A \to \R_+$ and a bijection $\sigma$ over $A_i$ or $N$, we let $\sigma w(\vec a) =w(\sigma \vec a)$.
\begin{definition}\label{def:p-symmetry}
Given a game $\cal G = \tup{N,A,v}$ and a weight function $w:A \to \R$, we say that an influence measure $\phi$ is {\em state-symmetric} with respect to $w$ (Sym-$w$) if for any permutation $\sigma:A_i \to A_i$, and all $j \in N$, $\phi_j(\sigma \cal G,\sigma w) = \phi_j(\cal G,w)$. That is, relabeling the states and letting them keep their original distributions does not change the value of any feature. Similarly, we say that an influence measure $\phi$ is {\em feature-symmetric} if for any permutation $\sigma: N \to N$, $\phi_{\sigma(i)}(\sigma\cal G,\sigma w) = \phi_i(\cal G,w)$. That is, relabeling the coordinate of a feature does not change its value.
\end{definition}

\begin{theorem}\label{thm:attvalue-prior}
If a probabilistic influence measure $\phi$ satisfies (D), (Sym) and (DU) with respect to some $\cal D$, then 
$$\phi_i(\cal G,\cal D) = C\pattvalue(\cal G,\cal D).$$
\end{theorem}
Before we proceed, we wish to emphasize two important aspects of Theorem~\ref{thm:attvalue-prior}. First, if we set $p(\vec a) = \frac{1}{|A|}$ then we obtain Theorem~\ref{thm:attvalue}. In other words, $\attvalue$ is an influence measure that assumes that all elements in the dataset are equally likely. 

Another point of note is the underlying process that the influence measures entail. If we assume that the weight function describes a distribution over $A$, one can think of the influence measure as the following process. We begin by picking a point from $A$ at random (uniformly at random in the case of $\attvalue$, and according to $w$ in Theorem~\ref{thm:attvalue-prior}); next, {\em fixing} the states of all other features, we measure the probability that $i$ can change the outcome, by sampling a different state according to the distribution $w(\cdot \mid \vec a_{-i})$.

Before we prove Theorem~\ref{thm:attvalue-prior}, let us prove the following lemma.
\begin{lemma}
$$\pattvalue(\cal G,w) = 2\sum_{\vec a_{-i} \in A_{-i}} w(\vec a_{-i})w(W_{\vec a_{-i}})w(L_{\vec a_{-i}})$$ 
\end{lemma}
\begin{proof}
\begin{align*}
\pattvalue(\cal G) =& \sum_{\vec a \in A}w(\vec a)\sum_{b \in A_i}w(b\mid \vec a_{-i})|v(\vec a_{-i},b) - v(\vec a)|\\
& = 2\sum_{\vec a_{-i} \in A}w(\vec a_{-i})\sum_{\substack{c \in A_i:\\v(\vec a_{-i},c) = 0}}\sum_{\substack{b \in A_i:\\v(\vec a_{-i},b) = 1}}w(c\mid \vec a_{-i})w(b\mid \vec a_{-i})\\
& = 2\sum_{\vec a_{-i} \in A}w(\vec a_{-i})\sum_{\substack{c \in A_i:\\v(\vec a_{-i},c) = 0}}w(c\mid \vec a_{-i})w(W_{\vec a_{-i}})\\
& = 2\sum_{\vec a_{-i} \in A}w(\vec a_{-i})w(L_{\vec a_{-i}})w(W_{\vec a_{-i}})\\
\end{align*}
\end{proof}  
\begin{lemma}
Let $f:\R^2 \to \R$ be a function that satisfies
\begin{enumerate}[(i)]
\item $f(x,0) = f(0,y) = 0$.
\item $f(x_1,y) + f(x_2,y) = f(x_1+x_2,y)$.
\item $f(x,y_1)+f(x,y_2) = f(x,y_1+y_2)$.
\end{enumerate}
Then there is some constant $c$ such that $f(x,y) = cxy$.
\end{lemma}
\begin{proof}
First, we show that $f(r x,y) = r f(x,y)$ for all $r \in \R$. Given any $n \in \N$, $f(nx,y) = nf(x,y)$ by property (2). Similarly, $f(\frac xn,y) = \frac1nf(x,y)$. Thus, for any rational number $q \in \Q$, we have $f(qx,y) = f(x,qy) = qf(x,y)$. Now, take any real number $r \in \R$. There exists a sequence of rational numbers $(q_n)_{n = 1}^\infty$ such that $\lim_{n \to \infty}q_n = r$. Thus, $f(rx,y) = \lim_{n \to \infty} f(q_nx,y) = \lim_{n \to \infty}q_nf(x,y) = rf(x,y)$ (and similarly $f(x,ry) = rf(x,y)$). 

Let us observe the partial derivatives of $f$ at $x \ne 0$:
\begin{align*}
\der{f}{x}(x^*,y^*) 	&= \lim_{\eps \to 0} \frac{f(x^*+\eps,y^*) - f(x^*,y^*)}{\eps}\\
						&= \lim_{\eps \to 0}\frac{(\frac{x^*+\eps}{x^*} - 1)f(x^*,y^*)}{\eps}= \frac{f(x^*,y^*)}{x^*}
\end{align*}
and similarly $\der{f}{y}(x^*,y^*) = \frac{f(x^*,y^*)}{y^*}$.
We obtain the following differential equation: $x\der fx - f= 0$. Its only solution is $f(x,y) = g(y)x + h(y)$. However, since $f(0,y) = 0$ for all $y$, we get that $h(y) \equiv 0$. Similarly, $f(x,y) = k(x)y$. Putting it all together, we get that $f(x,y) = cxy$. 

\end{proof}
\begin{lemma}\label{lem:attvalue-cxy}
If a value $\phi$ satisfies the (WDU) and (Sym-$w$) property, then it agrees with $\pattvalue$ on any game $\cal G = \tup{\{i\},A_i,v}$ with any weight function $w:A_i \to \R_+$
\end{lemma}
\begin{proof}
Let us write $W_i$ and $L_i$ to be the winning and losing states in $A_i$. By state symmetry we know that $\phi$ is only a function of $\left(w(b)\right)_{b \in W_i}$ and $\left(w(b)\right)_{b \in L_i}$. By the weighted disjoint union property, we know that 
$$\phi_i(\left(w(b)\right)_{b \in W_i},\left(w(b)\right)_{b \in L_i}) = \sum_{b \in W_i}\sum_{c \in L_i}\phi_i(w(b),w(c)).$$
Using the (WDU) property, we know that the following holds for single-feature games with only two states. Given $x_1,x_2,y \in \R_+$, the following holds:
\begin{align*}
\phi_i(x_1+x_2,y) =& \phi_i(x_1,y) + \phi_i(x_2,y)\\
\phi_i(y,x_1+x_2) =& \phi_i(y,x_1) + \phi_i(y,x_2)
\end{align*}
By Lemma~\ref{lem:attvalue-cxy}, we know that $\phi_i(x,y) = cxy = c\pattvalue_i(x,y)$. In particular, this implies that $\phi_i(\cal G,w) = \pattvalue_i(\cal G,w)$, and we are done. 
\end{proof}

\begin{proof}[Proof of Theorem~\ref{thm:attvalue-prior}]
First, we note that $\pattvalue$ satisfies (D), (Sym-$w$) and (WDU) (this is an easy exercise). We write $W$ to be the winning state vectors in $A$ and $L$ to be the losing state vectors in $A$. Now, if either $w(W) = 0$ or $w(L) = 0$, any influence measure that satisfies (D) assigns a value of zero to all $i \in N$, and the claim trivially holds. Thus, we assume that $w(W),w(L) > 0$. 

Next, according to the (DU) property, we can write 
$$\phi_i(W,L,w) = \sum_{\vec a_{-i} \in A_{-i}}\phi_i(W_{\vec a_{-i}},L_{\vec a_{-i}},w).$$
The argument is the same as the one used for the decomposition of $\attvalue$ in Theorem~\ref{thm:attvalue}. By the above lemmas, $\phi_i(W_{\vec a_{-i}},L_{\vec a_{-i}},w) = C\pattvalue_i(W_{\vec a_{-i}},L_{\vec a_{-i}},w)$. Note that by feature symmetry, it must be the case that the constant $C$ is independent of $i$. 
\end{proof}

\section{Generalized Distance Measures}\label{sec:distancemeasures}
Suppose that we have a set of feature vectors $B\subseteq A$. In previous sections we had assumed that there was some function $v:A \to \{0,1\}$ that classified a vector as either having a value of 0 or a value of 1. We then proceeded to provide an axiomatic characterization of influence measures in such settings. Influence was largely based on the following notion: a feature $i \in N$ can influence the vector $\vec a \in A$, if $|v(\vec a_{-i},b) - v(\vec a)| = 1$. 
Let us now consider a more general setting; instead of defining a classifier over data points, we have some semi-distance measure over the vectors. Recall that a pseudo-distance measure is a function $d:A\times A \to \R$ that satisfies all of the distance axioms, but $d(\vec a,\vec b) = 0$ does not necessarily imply that $\vec a = \vec b$. Given some pseudo-distance measure $d$ over $A$, rather than measuring influence by the measure $|v(\vec a_{-i},b) - v(\vec a)|$, we measure influence by $d((\vec a_{-i},b),\vec a)$. 

We observe that if $d(\vec a,\vec b) \in \{0,1\}$ for all $\vec a,\vec b \in A$, then we revert to the original setting. 

Given a pseudo-distance measure $d$ over $A$ and a dataset $B \subseteq A$, let us define $\cal P_d(B)$ to be the partition of $B$ into the equivalence classes defined by $\vec a \sim \vec b$ iff $d(\vec a,\vec b) = 0$. In other words, $\cal P_d(B)$ is the clustering of $B$ into points that are of equal distance to each other. 
Fixing a pseudo-distance $d$, we provide the following extensions of the axioms defined in Section~\ref{sec:featureinfluence}. 

We keep the notion of symmetry used in Section~\ref{sec:featureinfluence} (Definition~\ref{def:symmetry}): an influence measure satisfies symmetry if it is invariant under coordinate permutations, both for individual features (e.g. renaming males to females and vice versa should not change the influence of any feature), and between the features (e.g. renaming gender and age should not change feature influence). We do, however, adopt more general definitions of the dummy and disjoint union properties. 
\begin{definition}[$d$-Dummy]
We say that an influence measure satisfies the $d$-Dummy property if $\phi_i(B) = 0$ whenever $d((\vec a_{-i},b),\vec a)= 0$ for all $\vec a \in B$ and all $b \in A_i$ such that $(\vec a_{-i},b) \in B$.
\end{definition}
\begin{definition}[Feature Independence]
Let $B\subseteq A$ be a dataset, and let $B(\vec a_{-i}) = \{\vec b \in B\mid \vec b_{-i} = \vec a_{-i}\}$. An influence measure satisfies feature independence (FD) if 
$$\phi_i(B) = \sum_{\vec a_{-i} \in A_{-i}}\phi_i(B(\vec a_{-i})).$$
\end{definition}
\begin{definition}[$d$-Disjoint Union]
Let $B\subseteq A$ be a dataset, and let $\cal B = \{B_1,\dots,B_m\}$ be the equivalence classes of $B$ according to the pseudo-distance $d$. An influence measure $\phi$ satisfies the $d$-disjoint union, if for any $j \in \{1,\dots,m\}$, any partition $C,C'$ of $B_j$ satisifies
$$\phi_i(B_1,\dots,B_m) = \phi_i(\cal B_{-j},C) + \phi_i(\cal B_{-j},C') - \phi_i(\cal B_{-j}).$$
\end{definition}
Finally, the following axiom requires that in very minimal settings, a feature's influence should agree with $d$. 
\begin{definition}[Agreement with Distance]
\end{definition}
Given a dataset $B \subseteq A$, define 
\begin{align}
\attvalue_i^d(B) = \sum_{\vec a \in B}\sum_{b \in A_i: (\vec a_{-i},b) \in B} d((\vec a_{-i},b),\vec a)
\end{align}

\begin{lemma}\label{lem:single-feature-d-influence}
Let $B$ be a dataset of single-feature points. Then if $\phi$ satisfies, $d$-(D), $d$-(DU), (Sym), and (AD), then $\phi(B) = \attvalue^d(B)$ 
\end{lemma}
\begin{proof}[Proof Sketch]
We partition $B$ into its equivalence classes according to $d$, $\cal B = \{B_1,\dots,B_m\}$. In an argument similar to Lemma~\ref{lem:symmetry-dependence-on-winning-states}, we can show that the symmetry axiom implies that $\phi$ is a function of $|B_1|,\dots,|B_m|$. Let $w_j = |B_j|$; employing the $d$-disjoint union property and the dummy property, we obtain that there exists some $m \times m$ matrix $D'$ such that $\phi(B) = \vec w^TD'\vec w$, and $D'$ is 0 on the diagonal, non-negative, and symmetric (symmetry here is obtained via state symmetry).

To show that $D'$ must identify with the pseudo-distance, we employ the agreement with distance axiom on inputs to $\phi$ that have only two non-zero coordinates, to obtain the desired result.    
\end{proof}

\begin{theorem}
If an influence measure $\phi$ satisfies the $d$-dummy, $d$-disjoint union, symmetry and agreement with distance axioms, then 
$$\phi_i(B) = \alpha\sum_{\vec a \in B}\sum_{b \in A_i:(\vec a_{-i},b) \in B} d((\vec a_{-i},b),\vec a),$$
where $\alpha$ is a constant independent of $i$.
\end{theorem}
\begin{proof}[Proof Sketch]
The proof mostly follows the proof technique of Theorem~\ref{thm:attvalue}. 
Let us write the influence of $i$ under $d$ to be $\phi_i^d(A)$. 

Using the (FI) property, we decompose $\phi_i^d$ into $|A_{-i}|$ different single-feature datasets. Next, we apply Lemma~\ref{lem:single-feature-d-influence} on each of the datasets to show that identity holds.
\end{proof}
\fi

\end{document}